\documentclass[journal,twocolumn]{IEEEtran}
\usepackage{etex}

\usepackage{amsmath,amssymb,amsfonts}   
\usepackage{bm}                         
\usepackage{graphicx}                   
\usepackage{epstopdf}
\usepackage{psfrag}                     
\usepackage{float}
\usepackage{url}
\usepackage{cleveref}

\usepackage{amsmath,amsfonts,amssymb,mathrsfs,times,trfsigns}
\usepackage{color}
\usepackage{lastpage}
\usepackage{epsfig}
\usepackage{layout}
\usepackage{subcaption}
\usepackage{tikz}
\usepackage{pgfplots}
\pgfplotsset{compat=newest}
\usepgflibrary{arrows}
\pgfplotsset{
	every y tick label/.append style={font=\scriptsize},
	every x tick label/.append style={font=\scriptsize},
	every axis label/.append style={font=\scriptsize}
}
\newlength\figureheight
\newlength\figurewidth

\usepackage{mathtools,lipsum,cuted}
\usepackage{cite}
\usepackage[T1]{fontenc}
\usepackage[utf8]{inputenc}
\usepackage{authblk}
\usepackage{booktabs}

\usepackage{algorithm,multicol}
\usepackage[noend]{algpseudocode}

\makeatletter
\def\BState{\State\hskip-\ALG@thistlm}
\makeatother

\usepackage{amsthm} 

\usepackage[intoc]{nomencl}

\makenomenclature
\usepackage{filecontents}
\usepackage{datatool}
\usepackage{longtable}

\usepackage[english]{babel}
\usepackage{pgf,tikz}
\usepackage{tikzscale}
\usetikzlibrary{angles,calc,intersections,quotes,arrows.meta}
\usepackage{tkz-euclide}
\usetkzobj{all}
\usepackage{mathrsfs}
\usetikzlibrary{babel}

\usetikzlibrary{spy}

\definecolor{c1}{rgb}{0,0,0}%
\definecolor{c2}{rgb}{0,0,0}%
\definecolor{c3}{rgb}{0,0,0}%
\definecolor{c4}{rgb}{0,0,0}%
\definecolor{c5}{rgb}{0,0,0}%
\definecolor{c6}{rgb}{0,0,0}%
\definecolor{c7}{rgb}{0,0,0}%
\definecolor{c11}{rgb}{0,0,0}%
\definecolor{c22}{rgb}{0,0,0}%
\definecolor{c33}{rgb}{0,0,0}%
\definecolor{c44}{rgb}{0,0,0}%
\definecolor{c55}{rgb}{0,0,0}%
\definecolor{c66}{rgb}{0,0,0}%
\definecolor{c77}{rgb}{0,0,0}%
\definecolor{c77}{rgb}{0,0,0}%
\definecolor{c88}{rgb}{0,0,0}%
\definecolor{c99}{rgb}{0,0,0}%
\definecolor{c111}{rgb}{0,0,0}%
\newtheorem{theorem}{Theorem}

\newcommand{\ma}[1]{\mbox{\boldmath$#1$}}
\newcommand{\ten}[1]{\ma{\mathcal #1}}
\newcommand{\compl}{\mathbb{C}}
\newcommand{\real}{\mathbb{R}}
\DeclareMathOperator*{\argmin}{argmin}

\IEEEoverridecommandlockouts

\makeatletter
\def\footnoterule{\relax%
	\kern-5pt
	\hbox to \columnwidth{\hfill\vrule width 1.0\columnwidth height 0.4pt\hfill}
	\kern4.6pt}
\makeatother

\begin{document}
\title{First-Order Perturbation Analysis of the SECSI Framework for the Approximate CP Decomposition of 3-D Noise-Corrupted  Low-Rank Tensors}

\author{Sher Ali Cheema,~\IEEEmembership{Student~Member,~IEEE,} Emilio Rafael Balda, Yao Cheng,~\IEEEmembership{Member,~IEEE,} \newline Martin Haardt,~\IEEEmembership{Senior~Member,~IEEE},   Amir Weiss,~\IEEEmembership{Student~Member,~IEEE,} and \newline Arie Yeredor,~\IEEEmembership{Senior~Member,~IEEE}
\thanks{S. A. Cheema, E. R. Balda, Y. Cheng, and M.Haardt are with the Communication Research Laboratory, Ilmenau University of Technology, Germany
      (e-mail: sher-ali.cheema@tu-ilmenau.de, emilio.balda@tu-ilmenau.de, y.cheng@tu-ilmenau.de, and martin.haardt@tu-ilmenau.de) \\
      A. Weiss and A. Yeredor are with the School of Electrical Engineering, Tel-Aviv University, Israel
      (e-mail: amirwei2@mail.tau.ac.il, arie@eng.tau.ac.il)}}

\maketitle

\begin{abstract}
The Semi-Algebraic framework for the approximate  Canonical Polyadic (CP) decomposition via SImultaneaous matrix diagonalization (SECSI) is an efficient tool for the computation of the CP decomposition. The SECSI framework \textcolor{c1}{reformulates} the CP decomposition into a set of joint eigenvalue decomposition (JEVD) problems. Solving all JEVDs, we obtain multiple estimates of the factor matrices and the best estimate is chosen in a subsequent step by \textcolor{c55}{using an} exhaustive search or some heuristic \textcolor{c55}{strategy that reduces the computational complexity}. Moreover, the SECSI framework retains the option of choosing the number of JEVDs to be solved, thus providing an adjustable complexity-accuracy trade-off. In this work, we provide an analytical performance analysis of the SECSI framework \textcolor{c1}{for the computation of the approximate CP decomposition of a noise corrupted low-rank tensor}, where we \textcolor{c77}{derive} closed-form expressions of the relative mean square error for each of the estimated factor matrices. These expressions are \textcolor{c77}{obtained} using a first-order perturbation analysis and are formulated in terms of the second-order moments of the noise, such that apart from a zero mean, no assumptions on the noise statistics are required. Simulation results exhibit an excellent match between the obtained closed-form expressions and the empirical results. Moreover, we propose a new Performance Analysis \textcolor{c77}{based} Selection (PAS) scheme to choose the \textcolor{c77}{final factor matrix} estimate. \textcolor{c55}{The results show that} the proposed PAS scheme outperforms the existing \textcolor{c66}{heuristics}, especially in \textcolor{c77}{the} high SNR \textcolor{c66}{regime}. \looseness=-1
\end{abstract}
\IEEEpeerreviewmaketitle
\begin{IEEEkeywords}
Perturbation analysis, higher-order singular value decomposition (HOSVD), tensor signal processing.
\end{IEEEkeywords}
%
%
%
\section{Introduction}\label{ch:SECSI}
%
The Canonical Polyadic (CP) decomposition of $R$-way arrays is a powerful tool in multi-linear algebra. It allows to decompose a tensor into a sum of rank-one components. There exist many applications where the underlying signal of interest can be represented by a trilinear or multilinear CP model. These range from psychometrics and chemometrics over array signal processing and communications to biomedical signal processing, image compression or numerical mathematics \cite{CP_app1, CP_app2,CP_app3}. In practice, the signal of interest in these applications is contaminated by the noise. \textcolor{c66}{Therefore,} we only compute an approximate CP decomposition of the noisy signal. \looseness=-1

Algorithms \textcolor{c77}{for the computation of} an approximate CP decomposition from noisy observations are often based on Alternating Least Squares (ALS). These algorithms compute \textcolor{c55}{the} CP decomposition in an iterative manner procedure \cite{ALS1, ALS2}. The main drawbacks of ALS-based algorithms is that the number of required iterations may be very large\textcolor{c66}{, and} \textcolor{c55}{convergence is not guaranteed}. Moreover, ALS based algorithms are less accurate in ill-conditioned scenarios, \textcolor{c55}{especially if} the columns of the factor matrices are highly correlated.  Alternatively, semi-algebraic solutions, where the CP decomposition is rephrased into a generic problem such as \textcolor{c55}{the} Joint Eigenvalue Decomposition (JEVD) (also \textcolor{c55}{called} Simultaneous Matrix Diagonalization (SMD)),
have been proposed in the literature \cite{CP_app2}. The link between \textcolor{c88}{the} CP decomposition and \textcolor{c55}{the} JEVD is discussed in \cite{CPvsSMD} where it has been shown that the canonical components can be obtained from a simultaneous matrix diagonalization by congruence. A SEmi-algebraic framework for CP decompositions via SImultaneous matrix diagonalization (SECSI) was presented in \cite{SECSI_first, SECSI02, SECSI} for $R=3$ dimensional tensors \textcolor{c66}{and} was extended for tensors with $R > 3$ dimensions using the concept of generalized unfoldings (SECSI-GU) in \cite{SECSI3}.
The SECSI concept facilitates a distributed implementation on a parallel JEVDs to be solved depending upon the accuracy and the computational complexity requirements of the system. \textcolor{c77}{By solving} all JEVDs, multiple estimates of the factor matrices are obtained. The selection of the best factor matrices from the resulting estimates can be obtained either by using \textcolor{c66}{an} exhaustive search based best matching scheme or by using heuristic selection schemes with a reduced computational complexity. Several schemes with different accuracy-complexity trade-off points are presented in \cite{SECSI}. Thus, the SECSI framework results in more reliable estimates and also offers a flexible accuracy-complexity trade-off.\looseness=-1

\textcolor{c55}{An analytical performance assessment of the semi-algebraic algorithms to compute an approximate CP decomposition is of \textcolor{c88}{considerable} research interest. In \textcolor{c66}{the} literature, the performance of the CP decomposition is often evaluated using Monte-Carlo simulations.
To the best of our knowledge, there exists no analytical performance analysis of \textcolor{c66}{an approximate CP decomposition of noise-corrupted low-rank tensors} in the literature.} In this work, a first-order perturbation analysis of the SECSI framework is carried out, where apart from zero-mean and finite second order moments, no assumptions about the noise are required. The SECSI framework performs three distinct step to compute the approximate CP decomposition of a noisy tensor, \textcolor{c88}{as summarized in Fig. \ref{overview_SECSI}}. \textcolor{c66}{First, the} truncated higher order singular value decomposition (HOSVD) is \textcolor{c55}{used} to suppress \textcolor{c55}{the} noise. \textcolor{c66}{In the second step, several JEVDs are constructed from the core tensor.} \textcolor{c55}{This results in several estimates of the factor matrices. \textcolor{c66}{Lastly, the best factor \textcolor{c77}{matrices are} selected from these estimates by applying the}} best matching scheme or \textcolor{c66}{an appropriate heuristic scheme that has a lower computational complexity \cite{SECSI}.} Hence, a perturbation analysis for each of the steps is required for the overall performance analysis of the SECSI framework. In \cite{Asilmor_16}, we have already presented \textcolor{c55}{a} first-order perturbation analysis of low-rank tensor approximations based on the truncated HOSVD. We have also performed the perturbation analysis of JEVD algorithms \textcolor{c55}{which are based on \textcolor{c66}{the} indirect least squares (LS) cost function} in  \cite{Icassp_17}. In this \textcolor{c77}{paper}, we extend our work to the overall performance analysis of \textcolor{c66}{the} SECSI framework for $3$-D tensors.
\textcolor{c55}{Finally}, we present closed-form expressions for the relative Mean Square Factor Error (rMSFE) \textcolor{c77}{for each of the estimates of the three factor matrices}. These expressions are asymptotic in the SNR and \textcolor{c7}{are expressed} in terms of \textcolor{c7}{the} covariance matrix of the noise. Furthermore, we devise a new heuristic approach based on the performance analysis results to select the best estimates that we \textcolor{c55}{call} Performance Analysis based Selection (PAS) scheme. \looseness=-1
\begin{figure*}[t!]
          \centering
          \includegraphics[width=.95\linewidth]{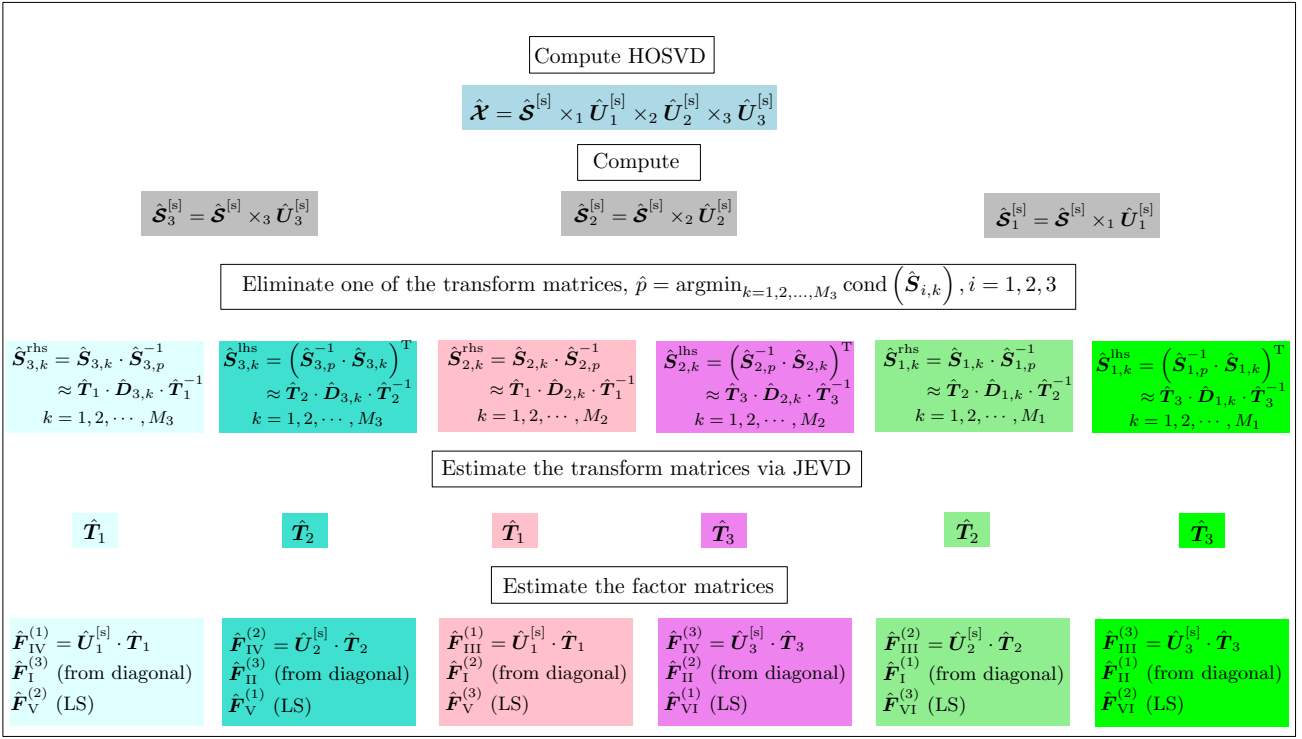}
          \caption{Overview of the SECSI framework to compute an approximate CP decomposition of a noise-corrupted low-rank tensor.}
          \vspace*{-0.6cm}
          \label{overview_SECSI}
  \end{figure*}	

\textcolor{c55}{The remainder of the paper is organized as follows. In Section \ref{overview}, we provide the data model.
We perform the first-order perturbation analysis of the SECSI framework in terms of known noisy tensor in Section \ref{sec:SECSI_perturbation}. The closed-form rMSFE expressions for each of the factor matrices \textcolor{c66}{for the first JEVD (corresponding to the first column of Fig. \ref{overview_SECSI})} are presented in Section \ref{sec:SECSI_Performance}. The results are extended for each of the factor matrices \textcolor{c66}{resulting from the second JEVD (corresponding to the second column of Fig. \ref{overview_SECSI})} in Section \ref{sec:SECSI_LHS}. \textcolor{c66}{The results of the remaining JEVDs in Fig. \ref{overview_SECSI}  are obtained via \textcolor{c77}{permutations} of the previous results that are discussed in Section \ref{sec:SECSI_Performance}.} In Section \ref{sec:SECSI_PAS}, we propose a new estimates selection scheme \textcolor{c66}{that} is based on the performance analysis results. The simulation results are discussed in Section \ref{sim} \textcolor{c66}{and the conclusions are provided in Section} \ref{conc}.}\looseness=-1

\textit{Notation}:
For the sake of notation, we use $a$, $\mathbf{a}$, $\ma{A}$, and $\ten{A}$ for a scalar, column vector, matrix, and tensor, respectively, where, $\ten{A}(i,j,k)$  defines the element $(i,j,k)$ of a tensor $\ten{A}$. The same applies \textcolor{c77}{to a} matrix $\ma{A}(i,j)$ and \textcolor{c77}{a} vector $\mathbf{a}(i)$. The superscripts $^{-1}$, $^{+}$, $^{*}$, $^{\mathrm{T}}$, and $^{\mathrm{H}}$  denote the matrix inverse, Moore-Penrose pseudo inverse, conjugate, transposition, and conjugate transposition, respectively. We also use the notation $\mathbb{E}\{\cdot\}$, $\mathrm{tr}\{\cdot\}$, $\otimes$, $\diamond$, $\|\cdot\|_{\mathrm{F}}$, and $\|\cdot\|_{2}$ for the expectation, trace, Kronecker product, the Khatri-Rao (column-wise Kronecker) product, Frobenius norm, and 2-norm operators, respectively. \textcolor{c88}{Moreover, we define an operator $\mathrm{stack}\{\cdot\}$ that arranges the $K$ vectors or matrices as
\begin{equation}
\mathrm{stack}\{\ma{A}^{(k)}\} =  \begin{bmatrix}
\ma{A}_1  \\
\vdots \\
\ma{A}_{K}  \\
\end{bmatrix} \;\; \forall k =1,2,\cdots,K
\end{equation}}

For a matrix $\ma{A} = [\mathbf{a}_1, \mathbf{a}_2, \dots, \mathbf{a}_M] \in \compl^{N \times M}$, the operator $\mathrm{vec}\{\cdot\}$ defines the vectorization operation  as
$ \mathrm{vec}\{\ma{A}\}^{\mathrm{T}} = [\mathbf{a}_1^{\mathrm{T}} , \mathbf{a}_2^{\mathrm{T}} , \dots, \mathbf{a}_M^{\mathrm{T}} ]$.
This operator has the property
\begin{equation}\label{eq:VecOp}
\mathrm{vec}\{\ma{A} \cdot \ma{X} \cdot \ma{B}\} = (\ma{B}^{\mathrm{T}} \otimes \ma{A})\cdot\mathrm{vec}\{\ma{X}\},
\end{equation}
where $\ma{A}$, $\ma{X}$, $\ma{B}$ are matrices \textcolor{c77}{with proper dimensions}. Moreover, we define the diagonalization operator $\mathrm{diag}(\cdot)$ as in Matlab. Note that, when this operator is applied to a vector, the output is a diagonal matrix, and when applied to a matrix, the output is a column vector. Therefore, for the sake of clarity, we use the notation $\mathrm{Diag}(\cdot)$ when applying this operator to a vector, and $\mathrm{diag}(\cdot)$ when applying it to a matrix. Furthermore we define the operators $\mathrm{Ddiag}(\cdot)$ and $\mathrm{Off}(\cdot)$ as
\begin{align*}
\mathrm{Ddiag}( \ma{X} ) &= \mathrm{Diag}( \mathrm{diag}(\ma{X})  ) \textcolor{c7}{ \in \compl^{N \times N}}\\
\mathrm{Off}( \ma{X} ) &= \ma{X} - \mathrm{Diag}( \mathrm{diag}(\ma{X})  ) \textcolor{c7}{ \in \compl^{N \times N}}\, ,
\end{align*}
where $\ma{X}$ is a \textcolor{c7}{square matrix matrix of size $N \times N$}. Note that the $\mathrm{Ddiag}(\cdot)$ operator sets all the off-diagonal elements of $\ma{X}$ to zero, while the $\mathrm{Off}(\cdot)$ operator sets the diagonal elements of $\ma{X}$ to zero. \looseness=-1

Let $\ten{A}\in \compl^{M_1 \times M_2 \times \cdots \times M_R}$ be a $R$-way tensor, where $M_{\text{$r$}}$ is the size along the $r$-th dimension and $[\ten{A}]_{(r)}$ denote the \textcolor{c7}{$r$-mode} unfolding of $\ten{A}$ which is performed according to the reverse cyclical order \cite{LathauwerPerfAnalysis}.  The $r$-mode product of a tensor $\ten{A}$ with a matrix $\ma{B} \in \compl^{N \times M_{\text{$r$}}}$ (i.e., $\ten{A} \times_{\text{$r$}} \ma{B}$) is defined as
\begin{equation*}
\ten{C} = \ten{A} \times_{\text{$r$}} \ma{B} \Longleftrightarrow  [\ten{C}]_{(\text{$r$})} = \ma{B} \cdot [\ten{A}]_{(\text{$r$})},
\end{equation*}
where $\ten{C}$ is a tensor with the corresponding dimensions. Moreover, if $\ten{C} = \ten{A} \times_{1} \ma{X}^{(1)} \times_{2} \ma{X}^{(2)} \times_3 \cdots \times_{R} \ma{X}^{(R)}$, then
\begin{align*}
&\left[\ten{C}\right]_{(r)} = \ma{X}^{(r)} \cdot [\ten{A}]_{(r)} \cdot \nonumber\\
&\left( \ma{X}^{(r+1)} \otimes \ma{X}^{(r+2)}\otimes \cdots
\ma{X}^{(R)} \otimes \ma{X}^{(1)} \otimes \cdots  \otimes \ma{X}^{(r-1)} \right)^{\mathrm{T}},
\end{align*}
where $\ma{X}^{(r)}, \forall r=1,2,\dots, R$ are matrices with the corresponding dimensions. \textcolor{c7}{For the sake of notational simplicity, we define the following notation for $r$-mode products
\begin{align*}
\ten{A}  \,\overset{R}{\underset{r=1}{ {\bigtimes}_r } } \ma{X}^{(r)} &= \ten{A} \times_1 \ma{X}^{(1)} \times_2 \ma{X}^{(2)} \times_3 \cdots \times_R \ma{X}^{(R)}.
\end{align*}
In addition,  $\|\cdot\|_{\mathrm{H}} $ \textcolor{c77}{denotes} the higher order norm of a tensor, defined as
\begin{equation*}
\| \ten{A}\|_{\mathrm{H}} = \| [\ten{A}]_{(\text{$r$})} \|_{\mathrm{F}} = \| \mathrm{vec}\{ [\ten{A}]_{(r)} \} \|_{2} \quad \forall r=1,2,\dots,R.
\end{equation*}
Moreover}, the space spanned by the $r$-mode vectors is termed $r$-space of $\ten{A}$, \textcolor{c77}{the rank of $[\ten{A}]_{({r})}$ is the $r$-rank of $\ten{A}$}. Note that in general, the $r$-ranks (also referred to as the multilinear ranks) of a tensor $\ten{A}$ can all be different. \textcolor{c88}{Furthermore, the tensor rank  refers to the smallest possible $r$ such that a tensor can be written as the sum of $r$ rank-one tensors. Note that the tensor rank is not directly related to the multilinear ranks.} Furthermore, $\ma{e}_{d, k}$ denotes the $k$-th standard basis column vector of size $d$.
\section{Data Model}
\label{overview}
One of the major challenges in data-driven applications is that only a noise-corrupted version of the noiseless tensor $\ten{X}_0 \in \compl^{M_1 \times M_2 \times M_3}$ of tensor rank $d$ is observed. \textcolor{c11}{Let us consider the non-degenerate case first }where $d\leq \text{min}\{M_1, M_2, M_3\}$. However, as discussed in Section \ref{non-degenerate}, the SECSI framework can also be applied \textcolor{c11}{in the} degenerate case where the tensor rank of the noiseless tensor may be greater than any one of the dimensions (i.e., $ M_1 < d \leq  \text{min} \{M_2, M_3\}$). The CP decomposition of \textcolor{c22}{such a low-rank}  noiseless tensor is given by
\begin{align}
\ten{X}_0 &= \ten{I}_{3,d} \times_1 \ma{F}^{(1)} \times_2 \ma{F}^{(2)} \times_3 \ma{F}^{(3)}, \label{cp1}
\end{align}
where $\ma{F}^{(r)} \in \compl^{M_r \times d}, \forall r=1,2,3$ is the factor matrix in the $r$-th mode and $\ten{I}_{3,d}$ is the $3$-way identity tensor of size $d \times d \times d$.
\textcolor{c99}{In this work, we assume that the factor matrices are known, and the goal of the SECSI framework is to estimate them.}
%
%
\textcolor{c44}{For future reference}, the SVD of the $r$-mode unfolding of the noiseless low-rank tensor $\ten{X}_0 \in \compl^{M_1 \times M_2 \times M_3}$ is given as
\begin{align}
&[\ten{X}_0]_{(r)} =
\ma{U}_{r} \cdot \ma{\Sigma}_{r} \cdot \ma{V}_{r}^{\mathrm{H}} \nonumber\\
& = \left[ \begin{matrix} \ma{U}_{r}^{[\mathrm{s}]} & \ma{U}_{r}^{[\mathrm{n}]} \end{matrix} \right] \!\!
\left[ \begin{matrix} \ma{\Sigma}_{r}^{[\mathrm{s}]} & \ma{0}_{d \times M_{r}}  \\
       \ma{0}_{(M_{r}-d)\times d} & \ma{0}_{(M_{r}-d)\times M_{r}} \end{matrix} \right] \!\!
\left[ \begin{matrix} \ma{V}_{r}^{[\mathrm{s}]} & \ma{V}_{r}^{[\mathrm{n}]} \end{matrix} \right]^{\mathrm{H}} \!\!\!\nonumber\\
&= \ma{U}_{r}^{[\mathrm{s}]}  \cdot \ma{\Sigma}_{r}^{[\mathrm{s}]}  \cdot \ma{V}_{r}^{[\mathrm{s}]^{\mathrm{H}}}, \textcolor{c11}{\forall r=1,2,3}\label{r_mode_svd}
\end{align}
where the superscripts  $[\mathrm{s}]$ and $[\mathrm{n}]$ represent \textcolor{c11}{the} signal and \textcolor{c11}{the} noise subspaces, respectively. Let
\begin{equation}
 \ten{X} = \ten{X}_0 + \ten{N} \in \compl^{M_1 \times M_2 \times M_3}
\end{equation}
 be the observed noisy tensor where the desired signal component $\ten{X}_0$ is superimposed by a zero-mean additive noise tensor $\ten{N} \in \compl^{M_1 \times M_2 \times M_3}$. The SVD of the $r$-mode unfolding of the observed noisy tensor $\ten{X}$ is given as
\begin{align}
&[\ten{X}]_{(r)} =
\hat{ \ma{U}}_{r} \cdot \hat{ \ma{\Sigma}}_{r} \cdot \hat{\ma{V}}_{r}^{\mathrm{H}} \nonumber\\
& = \left[ \begin{matrix} \hat{ \ma{U}}_{r}^{[\mathrm{s}]} & \hat{ \ma{U}}_{r}^{[\mathrm{n}]} \end{matrix} \right] \!\!
\left[ \begin{matrix} \hat{ \ma{\Sigma}}_{r}^{[\mathrm{s}]} & \ma{0}_{d \times M_{r}} \\
       \ma{0}_{(M_{r}-d)\times d} & \hat{ \ma{\Sigma}}_{r}^{[\mathrm{n}]}\end{matrix} \right] \!\!
\left[ \begin{matrix} \hat{ \ma{V}}_{r}^{[\mathrm{s}]} & \hat{ \ma{V}}_{r}^{[\mathrm{n}]} \end{matrix} \right]^{\mathrm{H}}.  \label{noisy_svd}
\end{align}
\section{First-Order Perturbation Analysis}\label{sec:SECSI_perturbation}
\subsection{Perturbation of the Truncated HOSVD}
A low-rank approximation \textcolor{c44}{of $\ten{X}$ }can be computed by \textcolor{c44}{truncating the HOSVD of the noisy tensor \[\ten{X} =\hat{\ten{S}} \times_1 \hat{\ma{U}}_{1}\times_2 \hat{\ma{U}}_{2} \times_3 \hat{\ma{U}}_{3}\]} as \cite{LathauwerPerfAnalysis}
\begin{align}
\textcolor{c11}{\hat{\ten{X}}} = \hat{\ten{S}}^{[\mathrm{s}]} \times_1 \hat{\ma{U}}_{1}^{[\mathrm{s}]} \times_2 \hat{\ma{U}}_{2}^{[\mathrm{s}]} \times_3 \hat{\ma{U}}_{3}^{[\mathrm{s}]}, \label{T_hosvd}
\end{align}
where $\hat{\ten{S}}^{[\mathrm{s}]} \in \compl^{d \times d \times d}$ is the \textcolor{c11}{truncated} core tensor and $ \hat{\ma{U}}_{r}^{[\mathrm{s}]} \in \compl^{M_r \times d}, \forall r=1,2,3$ is \textcolor{c11}{obtained from \cref{noisy_svd}}.  In \cite{Asilmor_16}, we presented a first order perturbation analysis of the truncated HOSVD where we obtained analytical expressions for the signal subspace error in each dimension of \textcolor{c88}{the tensor. Additionally}, we also obtained  the analytical expressions for the tensor reconstruction error induced by the low-rank approximation of the noise corrupted tensor. \textcolor{c11}{Let us express the noisy estimates in \cref{T_hosvd} as}
\begin{align}
\hat{\ma{U}}_{r}^{[\mathrm{s}]} &\triangleq \ma{U}_{r}^{[\mathrm{s}]} + \Delta \ma{U}_{r}^{[\mathrm{s}]}, \quad \forall r=1,2,3 \label{eq:SECSI_Ur}\\
\hat{\ten{S}}^{[\mathrm{s}]} &\triangleq {\ten{S}}^{[\mathrm{s}]} + \Delta {\ten{S}}^{[\mathrm{s}]}. \label{eq:SECSI_S}
\end{align}
 The perturbation present in the $r$-mode signal subspace estimate $\hat{\ma{U}}_{r}^{[\mathrm{s}]}$ is given by \cite{SVD_per}
\begin{align}
\Delta \ma{U}_{r}^{[\mathrm{s}]} &= \ma{U}_{r}^{[\mathrm{n}]} \cdot \ma{U}_{r}^{[\mathrm{n}]^{\mathrm{H}}} \cdot [{\ten{N}}]_{(r)} \cdot \ma{V}_{r}^{[\mathrm{s}]} \cdot \ma{\Sigma}_{r}^{[\mathrm{s}]^{-1}} + \mathcal{O}(\Delta^2) \nonumber\\
&\textcolor{c7}{= \ma{\Gamma}_{r}^{[\mathrm{n}]} \cdot [{\ten{N}}]_{(r)} \cdot \ma{V}_{r}^{[\mathrm{s}]} \cdot \ma{\Sigma}_{r}^{[\mathrm{s}]^{-1}} + \mathcal{O}(\Delta^2)} \, , \label{eq:SECSI_Delta_Ur}
\end{align}
\textcolor{c11}{where $\ma{\Gamma}_{r}^{[\mathrm{n}]} \triangleq \ma{U}_{r}^{[\mathrm{n}]} \cdot \ma{U}_{r}^{[\mathrm{n}]^{\mathrm{H}}}$ and all higher order terms are contained in $\mathcal{O}(\Delta^2)$.} Furthermore, we use this result to expand the expression for the truncated core tensor $\hat{\ten{S}}^{[\mathrm{s}]} = {\ten{X}} \overset{3}{\underset{r=1}{ {\bigtimes}_r } } \hat{\ma{U}}_{r}^{[\mathrm{s}]^{\mathrm{H}}}$  as
\begin{align*}
\ten{S}^{[\mathrm{s}]} + \Delta \ten{S}^{[\mathrm{s}]} &= \left(\ten{X}_0 + \ten{N} \right) \overset{3}{\underset{r=1}{ {\bigtimes}_r } } \left(\ma{U}_{r}^{[\mathrm{s}]^{\mathrm{H}}} + \Delta\ma{U}_{r}^{[\mathrm{s}]^{\mathrm{H}}} \right) + \mathcal{O}(\Delta^2)\\
&= \ten{X}_0 \overset{3}{\underset{r=1}{ {\bigtimes}_r } } \ma{U}_{r}^{[\mathrm{s}]^{\mathrm{H}}} + \ten{N} \overset{3}{\underset{r=1}{ {\bigtimes}_r } } \ma{U}_{r}^{[\mathrm{s}]^{\mathrm{H}}}\\
&\quad + \ten{X}_0 \times _{1} \Delta \ma{U}_{1}^{[\mathrm{s}]^{\mathrm{H}}} \times _{2} \ma{U}_{2}^{[\mathrm{s}]^{\mathrm{H}}} \times _{3} \ma{U}_{3}^{[\mathrm{s}]^{\mathrm{H}}} \\
&\quad + \ten{X}_0 \times _{1}  \ma{U}_{1}^{[\mathrm{s}]^{\mathrm{H}}} \times _{2} \Delta\ma{U}_{2}^{[\mathrm{s}]^{\mathrm{H}}} \times _{3} \ma{U}_{3}^{[\mathrm{s}]^{\mathrm{H}}} \\
&\quad + \ten{X}_0 \times _{1} \ma{U}_{1}^{[\mathrm{s}]^{\mathrm{H}}} \times _{2} \ma{U}_{2}^{[\mathrm{s}]^{\mathrm{H}}} \times _{3} \Delta\ma{U}_{3}^{[\mathrm{s}]^{\mathrm{H}}}  + \mathcal{O}(\Delta^2).
\end{align*}
Note that all terms that include \textcolor{c88}{products of more than one "$\Delta$ term"} are included in $\mathcal{O}(\Delta^2)$. Moreover, $\ten{N}$ is also considered as a “$\Delta$ term”.
\textcolor{c22}{In the noiseless case, the} truncated core tensor $\ten{S}^{[\mathrm{s}]}$ is equal to $\ten{S}^{[\mathrm{s}]} \triangleq \ten{X}_0 \overset{3}{\underset{r=1}{ {\bigtimes}_r } } \ma{U}_{r}^{[\mathrm{s}]^{\mathrm{H}}}$. \textcolor{c111}{Using the definitions in \cref{r_mode_svd} and \cref{eq:SECSI_Delta_Ur}, the above expression simplifies to
\begin{align}
\Delta \ten{S}^{[\mathrm{s}]} &= \ten{N} \overset{3}{\underset{r=1}{ {\bigtimes}_r } } \ma{U}_{r}^{[\mathrm{s}]^{\mathrm{H}}} + \mathcal{O}(\Delta^2) \, . \label{eq:SECSI_DeltaSs}
\end{align}}
\begin{table*}[t!]
\centering
\begin{tabular}{|l|l|}
\hline
\multicolumn{2}{|l|}{\hspace{100pt} 1: Compute $\hat{\ma{U}}_{r}^{[\mathrm{s}]}$ for $r=1,2,3$ via the HOSVD of ${\ten{X}}$}  \\
\multicolumn{2}{|l|}{\hspace{100pt} 2: $\hat{\ten{S}}^{[\mathrm{s}]} = {\ten{X}} \overset{R}{\underset{r=1}{ {\bigtimes}_r } } \hat{\ma{U}}_{r}^{[\mathrm{s}]^{\mathrm{H}}}$}                                                                                                                                                  \\
\multicolumn{2}{|l|}{\hspace{100pt} 3: $\hat{\ten{S}}_3^{[\mathrm{s}]} = \hat{\ten{S}}^{[\mathrm{s}]} \times _{3} \hat{\ma{U}}_{3}^{[\mathrm{s}]}$}  \\
\multicolumn{2}{|l|}{\hspace{100pt} 4: $\hat{\ma{S}}_{3,k} = \hat{\ten{S}}_3^{[\mathrm{s}]} \times_{3} \ma{e}_{M_3,k}^{\mathrm{T}}$ for $k=1,2,\dots, M_3$}                                                                                                                                                                      \\
\multicolumn{2}{|l|}{\hspace{100pt} 5: $\textcolor{c7}{\hat{p}} = \argmin_{k = 1,2,\dots, M_3} \mathrm{cond}\left(\hat{\ma{S}}_{3,k}\right)$} \\ \hline
6: $\hat{\ma{S}}_{3,k}^{\mathrm{rhs}} = \hat{\ma{S}}_{3,k} \cdot,\hat{\ma{S}}_{3,p}^{-1}$ for $k=1,2,\dots,M_3$                                 & 6: $\hat{\ma{S}}_{3,k}^{\mathrm{lhs}} \triangleq \left(\hat{\ma{S}}_{3,p}^{-1} \cdot \hat{\ma{S}}_{3,k}\right)^{\mathrm{T}}$ for $k=1,2,\dots,M_3$              \\
7: Compute $\hat{\ma{T}}_1$ and $\hat{\ma{D}}_{3,k}$ via JEVD of $\left\{\hat{\ma{S}}_{3,k}^{\mathrm{rhs}}\right\}_{k=1}^{M_3}$                 & 7: Compute $\hat{\ma{T}}_2$ and $\hat{\ma{D}}_{3,k}$ via JEVD of $\left\{\hat{\ma{S}}_{3,k}^{\mathrm{lhs}}\right\}_{k=1}^{M_3}$                                  \\
8: $\hat{{\ma{F}}}^{(1)} = \hat{\ma{U}}_{1}^{[\mathrm{s}]} \cdot \hat{\ma{T}}_1$                                                                & 8: $\hat{{\ma{F}}}^{(1)} = \left[\ten{X}\right]_{(1)} \cdot \left[ \hat{{\ma{F}}}^{(2)} \diamond \hat{{ \ma{F}}}^{(3)}\right]^{+\mathrm{T}}$ \\
9: $\hat{\tilde{\ma{F}}}^{(3)}(k,:) = \mathrm{diag}\left( \hat{\ma{D}}_{3,k} \right)^{\mathrm{T}}$                                              & 9: $\hat{{\ma{F}}}^{(2)} = \hat{\ma{U}}_{2}^{[\mathrm{s}]}\cdot \hat{\ma{T}}_{2}$                                                                         \\
10: $\hat{{\ma{F}}}^{(2)} = [{\ten{X}}]_{(2)} \cdot \left( \hat{{\ma{F}}}^{(3)} \diamond \hat{\ma{F}}^{(1)} \right) ^{+\mathrm{T}}$ & 10: $\hat{\tilde{\ma{F}}}^{(3)}(k,:) = \mathrm{diag}\left( \hat{\ma{D}}_{3,k} \right)^{\mathrm{T}}$                                                             \\ \hline
\end{tabular}
\caption{SECSI Algorithm: Factor matrix estimates resulting from the 2 JEVD construction from Eq. (13). The whole SECSI framework is shown in Fig. \ref{overview_SECSI}}.
\label{alg:SECSI}
\vspace*{-0.5cm}
\end{table*}
\subsection{Perturbation of the JEVD Estimates }

For a $3$-way array, we can construct up to 6 JEVD problems in the SECSI framework \cite{SECSI} that are obtained \textcolor{c88}{from
\begin{align}
\hat{\ten{S}}^{[\mathrm{s}]}_{1} &= \hat{\ten{S}}^{[\mathrm{s}]} \times_1 \hat{\ma{U}}_{1}^{[\mathrm{s}]} \label{1_mode} \in \mathbb{C}^{M_1\times d\times d}\\
\hat{\ten{S}}^{[\mathrm{s}]}_{2} &= \hat{\ten{S}}^{[\mathrm{s}]} \times_2 \hat{\ma{U}}_{2}^{[\mathrm{s}]} \label{2_mode} \in \mathbb{C}^{d\times M_2\times d}\\
\hat{\ten{S}}^{[\mathrm{s}]}_{3} &= \hat{\ten{S}}^{[\mathrm{s}]} \times_3 \hat{\ma{U}}_{3}^{[\mathrm{s}]} \in \mathbb{C}^{d\times d\times M_3}, \label{3_mode}
\end{align}
respectively.} As an example, \textcolor{c11}{let us} consider the \textcolor{c44}{two} JEVD problems constructed from \cref{3_mode}. \textcolor{c11}{Here the} noisy estimate $\hat{\ten{S}}_3^{[\mathrm{s}]}$ \textcolor{c11}{can be expressed as}
\begin{align}
\hat{\ten{S}}_3^{[\mathrm{s}]} &\triangleq {\ten{S}}_3^{[\mathrm{s}]} + \Delta {\ten{S}}_3^{[\mathrm{s}]}. \label{eq:SECSI_S3}
\end{align}
 \textcolor{c44}{Using \cref{3_mode}, we get}
\begin{align}
\ten{S}_3^{[\mathrm{s}]} &+ \Delta \ten{S}_3^{[\mathrm{s}]} = \left(\ten{S}^{[\mathrm{s}]} + \Delta\ten{S}^{[\mathrm{s}]}\right) \times _{3}  \left(\ma{U}_{3}^{[\mathrm{s}]} + \Delta\ma{U}_{3}^{[\mathrm{s}]} \right) \nonumber \\
&=\ten{S}_3^{[\mathrm{s}]} + \Delta\ten{S}^{[\mathrm{s}]} \times _{3} \ma{U}_{3}^{[\mathrm{s}]} + \ten{S}^{[\mathrm{s}]} \times _{3} \Delta \ma{U}_{3}^{[\mathrm{s}]} + \mathcal{O}(\Delta^2). \nonumber
\end{align}
\textcolor{c11}{Therefore, we have}
\begin{align}
\Delta \ten{S}_3^{[\mathrm{s}]} =& \Delta \ten{S}^{[\mathrm{s}]} \times _{3} \ma{U}_{3}^{[\mathrm{s}]} + \ten{S}^{[\mathrm{s}]} \times _{3} \Delta \ma{U}_{3}^{[\mathrm{s}]} + \mathcal{O}(\Delta^2).\label{eq:SECSI:Delta_S3}
\end{align}
\textcolor{c11}{According to \cite{SECSI_first, SECSI}}, we define the 3-mode slices  \textcolor{c44}{of $\hat{\ma{S}}_{3}$} as \textcolor{c88}{
\begin{align}
\hat{\ma{S}}_{3,k} \triangleq \hat{\ten{S}}^{[\mathrm{s}]}_3 \times_3 \ma{e}_{M_3, k}^{\mathrm{T}} \in \mathbb{C}^{d\times d}, \quad k=1,2,...M_3, \label{slices}
\end{align}
where}
$ \hat{\ma{S}}_{3,k}$ represents the $k$-th slice (along the third dimension) of $\hat{\ten{S}}_3^{[\mathrm{s}]}$ \textcolor{c22}{in \cref{3_mode}}. As explained in \cite{SECSI}, these slices satisfy
\begin{equation}
\mathrm{Diag}\left\{  \hat{\ma{F}}^{(3)}(k,:) \right\} \approx \hat{\ma{T}}_{1}^{-1} \cdot \hat{\ma{S}}_{3,k} \cdot \hat{\ma{T}}_{2}^{-1}, \label{k-slice-def}
\end{equation}
\textcolor{c11}{where $\hat{\ma{T}}_{1}$ and $\hat{\ma{T}}_{2}$ are transformation matrices \textcolor{c88}{obtained by solving the associated JEVD problems}. In the noiseless case, the factor matrices are related to the signal subspaces via these transformation matrices as $\ma{F}^{(r)} = {\ma{U}}_{r}^{[\mathrm{s}]}\cdot \ma{T}_{r}$, \textcolor{c88}{$ r =1,2,3$}.} \textcolor{c11}{Defining} the perturbation in the $k$-th slice, \textcolor{c11}{we get}
\begin{align}
\hat{\ma{S}}_{3,k} &\triangleq {\ma{S}}_{3,k}  + \Delta {\ma{S}}_{3,k}, \quad \forall k=1,2,\dots, M_3. \label{eq:SECSI_S3k}
\end{align}
\textcolor{c44}{Using \cref{slices}, this results in}
\begin{align}
\Delta \ma{S}_{3,k} =& \Delta \ten{S}_3^{[\mathrm{s}]} \times _{3} \ma{e}_{M_3, k}^{\mathrm{T}} + \mathcal{O}(\Delta^2). \label{eq:Delta_S3k}
\end{align}
\textcolor{c22}{According to \cite{SECSI},}  we select the slice of $\hat{\ten{S}}_{3}^{[\mathrm{s}]}$ with the lowest condition number, \textcolor{c22}{i.e.,} $\hat{\ma{S}}_{3,p}$ where $p = \argmin_{k}\left\{ \mathrm{cond}\left( \hat{\ma{S}}_{3,k} \right) \right\}$ and $\mathrm{cond}(\cdot)$ denotes the condition number operator.
This leads to two \textcolor{c44}{sets of matrices}, namely the right-hand-side (rhs) set and the left-hand-side (lhs) \textcolor{c22}{set that} are defined as
\begin{align}
\hat{\ma{S}}_{3,k}^{\mathrm{rhs}} &\triangleq \hat{\ma{S}}_{3,k} \cdot \hat{\ma{S}}_{3,p}^{-1}, \quad \forall k=1,2,\dots, M_3 \label{eq:RHS1}\\
\hat{\ma{S}}_{3,k}^{\mathrm{lhs}} &\triangleq \left(\hat{\ma{S}}_{3,p}^{-1} \cdot \hat{\ma{S}}_{3,k}\right)^{\mathrm{T}}, \quad \forall k=1,2,\dots, M_3. \label{eq:LHS1}
\end{align}
As an example, we compute the perturbation in \cref{eq:RHS1}. \textcolor{c22}{To this end, let us obtain the perturbation in $\hat{\ma{S}}_{3,k}^{\mathrm{rhs}}$ as}
\begin{align}
\hat{\ma{S}}_{3,k}^{\mathrm{rhs}} &\triangleq {\ma{S}}_{3,k}^{\mathrm{rhs}} + \Delta {\ma{S}}_{3,k}^{\mathrm{rhs}}, \quad \forall k=1,2,\dots, M_3. \label{eq:SECSI_S3kRHS}
\end{align}
Using this definition, we now expand \cref{eq:RHS1}, as
\begin{align*}
\ma{S}_{3,k}^{\mathrm{rhs}} + \Delta \ma{S}_{3,k}^{\mathrm{rhs}}
= \left(\ma{S}_{3,k} + \Delta\ma{S}_{3,k}\right) \left(\ma{S}_{3,p} + \Delta\ma{S}_{3,p}\right)^{-1} + \mathcal{O}(\Delta^2)
\end{align*}
Using the Taylor's expansion  \textcolor{c44}{of the matrix inverse}, we get
\begin{align*}
\left(\ma{S}_{3,p} + \Delta\ma{S}_{3,p}\right)^{-1}= \ma{S}_{3,p}^{-1}-
\ma{S}_{3,p}^{-1} \cdot \Delta \ma{S}_{3,p} \cdot \ma{S}_{3,p}^{-1}  + \mathcal{O}(\Delta^2).
\end{align*}
\textcolor{c22}{According to  \textcolor{c88}{\cref{eq:SECSI_S3kRHS}}, the perturbation in the slices $\hat{\ma{S}}_{3,k}^{\mathrm{rhs}}, \forall k=1,2,\dots,M_3$  is given by}
\begin{align}
\Delta \ma{S}_{3,k}^{\mathrm{rhs}}
= \Delta \ma{S}_{3,k} \cdot \ma{S}_{3,p}^{-1} - \ma{S}_{3,k} \cdot \ma{S}_{3,p}^{-1} \cdot \Delta \ma{S}_{3,p} \cdot \ma{S}_{3,p}^{-1}  + \mathcal{O}(\Delta^2). \label{eq:DeltaS3_RHS}
\end{align}
\textcolor{c22}{Using the results in \cref{k-slice-def}, it is easy to show that the two sets of matrices in \cref{eq:RHS1} and \cref{eq:LHS1} correspond to the following JEVD problems
\begin{align}
\hat{\ma{S}}_{3,k}^{\mathrm{rhs}} \approx \hat{\ma{T}}_{1} \cdot \hat{\ma{D}}_{3,k}\cdot {\ma{T}}_{1}^{-1} \quad \forall k=1,2,\dots, M_3, \label{eq:RHS}\\
\hat{\ma{S}}_{3,k}^{\mathrm{lhs}} \approx \hat{\ma{T}}_{2} \cdot \hat{\ma{D}}_{3,k}\cdot \hat{\ma{T}}_{2}^{-1}, \quad \forall k=1,2,\dots, M_3,  \label{eq:LHS}
\end{align}
respectively,} where  \textcolor{c44}{the diagonal matrices} $\hat{\ma{D}}_{3,k} $ are defined as\textcolor{c88}{
\begin{align}
\hat{\ma{D}}_{3,k} \triangleq \mathrm{Diag}\left\{  \hat{\ma{F}}^{(3)}(k,:) \right\} \cdot \mathrm{Diag}\left\{  \hat{\ma{F}}^{(3)}(p,:) \right\}^{-1} \label{D3k}.
\end{align}
Eqs.} (\ref{eq:RHS}) and (\ref{eq:LHS}) show that $\hat{\ma{T}}_{1}$ and $\hat{\ma{T}}_{2}$ can be found via an approximate joint diagonalization of the matrix slices $\hat{\ma{S}}_{3,k}^{\mathrm{rhs}}$ and $\hat{\ma{S}}_{3,k}^{\mathrm{lhs}}$, respectively. \textcolor{c44}{Such an} approximate joint diagonalization of \cref{eq:RHS} and \cref{eq:LHS} can, for instance, be achieved via joint diagonalization algorithms based on the indirect least squares (LS) cost function such as Sh-Rt \cite{ShRt}, JDTM \cite{JDTM}, \textcolor{c44}{or the} coupled JEVD \cite{Coupled}.
In \cite{Icassp_17}, we have presented a first order perturbation analysis of JEVD algorithms that are based on the indirect LS cost function. We can use these results to obtain analytical expressions for the perturbation in the estimates $\hat{\ma{T}}_1$, $\hat{\ma{T}}_2$, and $\hat{\ma{D}}_{3,k}, \forall k=1,2,\dots,M_3$. To this end, the perturbations in the $\hat{\ma{T}}_{1}$, $\hat{\ma{T}}_{2}$, and $\hat{\ma{D}}_{3,k}$ estimates are defined as
\begin{align}
\hat{\ma{T}}_1 &\triangleq {\ma{T}}_1 + \Delta {\ma{T}}_1 \label{eq:SECSI_T1}\\
\hat{\ma{T}}_2 &\triangleq {\ma{T}}_2 + \Delta {\ma{T}}_2 \label{eq:SECSI_T2}\\
\hat{\ma{D}}_{3,k} &\triangleq {\ma{D}}_{3,k} + \Delta {\ma{D}}_{3,k}, \quad \forall k=1,2,\dots, M_3 \label{eq:SECSI_D3k}\,.
\end{align}
\textcolor{c22}{According to \cite{Icassp_17}}, we also define the following matrices
\begin{align}
\ma{B}_0 &= \ma{J}_{(d)} \cdot \left(\ma{T}_{1}^{\mathrm{T}} \otimes \ma{T}_{1}^{-1} \right)\\
\ma{s}_k &= \mathrm{vec}\left\{\Delta \ma{S}_{3,k}^{\mathrm{rhs}} \right\}\\
\ma{A}_k &= \ma{J}_{(d)} \cdot \left[\left(\ma{I}_M \otimes \ma{T}_{1}^{-1}\cdot \ma{S}_{3,k}^{\mathrm{rhs}} \right) - \left(\ma{D}_{3,k} \otimes \ma{T}_{1}^{-1} \right) \right] \, ,\label{eq:w_secsi}
\end{align}
where \textcolor{c88}{$\ma{J}_{(d)} \in \{0,1\}^{d^2\times d^2}$} is a selection matrix \textcolor{c44}{that satisfies} the relation $\mathrm{vec}\left\{\mathrm{Off}\left( \ma{X} \right) \right\} = \ma{J}_{(d)} \cdot \mathrm{vec}\left\{ \ma{X} \right\}$, for any given $\ma{X} \in \compl^{d \times d}$.
 \textcolor{c22}{After defining the quantities}
\begin{equation}\label{eq:BnAw_SECSI}
\ma{A} =  \begin{bmatrix}
\ma{A}_1  \\
\ma{A}_2  \\
\vdots \\
\ma{A}_{M_3}  \\
\end{bmatrix},
\quad
\bm{B} = \bm{I}_{M_3} \otimes \ma{B}_0,
\quad
\ma{s} =
\begin{bmatrix}
\ma{s}_1 \\
\ma{s}_2 \\
\vdots \\
\ma{s}_{M_3} \\
\end{bmatrix},
\end{equation}
we can use the results obtained in \cite{Icassp_17}. This leads to
\begin{align}
\mathrm{vec}\left\{\Delta\ma{T}_{1}\right\} &= -\ma{A}^{+} \cdot \ma{B} \cdot \ma{s} + \mathcal{O}(\Delta^2)  \label{eq:SECSI_vec_Delta_T1} \\
\Delta \ma{D}_{3,k} &=  \mathrm{Ddiag} \left( \ma{T}_{1}^{-1} \cdot \Delta \ma{S}_{3,k}^{\mathrm{rhs}} \cdot \ma{T}_{1} \right) + \mathcal{O}(\Delta^2), \, \label{eq:SECSI_vec_Delta_D3k}
\end{align}
for the rhs JEVD problem.
%
\subsection{Perturbation of the Factor Matrix Estimates }
\label{Pert_factor_matrix}
\textcolor{c22}{Using the SECSI framework \textcolor{c88}{for a 3-way tensor}, we can get up to six different estimates for each factor matrix. \textcolor{c44}{The factor matrix} $\ma{F}^{(1)}$ can be estimated from the transform matrix $\hat{\ma{T}}_1$ via $\hat{{\ma{F}}}^{(1)} = \hat{\ma{U}}_{1}^{[\mathrm{s}]}\cdot \hat{\ma{T}}_{1}$ by using the result of the JEVD defined in \cref{eq:RHS}}. Expanding this equation leads to
\begin{align*}
\hat{\ma{F}}^{(1)} &=  \left(\ma{U}_{1}^{[\mathrm{s}]} + \Delta \ma{U}_{1}^{[\mathrm{s}]}\right) \cdot\left( \ma{T}_{1} + \Delta \ma{T}_{1}\right) +\mathcal{O}(\Delta^2)\\
&= \ma{F}^{(1)} + \Delta \ma{U}_{1}^{[\mathrm{s}]} \cdot \ma{T}_{1} +  \ma{U}_{1}^{[\mathrm{s}]} \cdot \Delta \ma{T}_{1} + \mathcal{O}(\Delta^2).
\end{align*}
\textcolor{c77}{Again,} we express the perturbation in $\hat{{\ma{F}}}^{(1)}$ as
\begin{equation}\label{eq:SECSI_F1}
\hat{{\ma{F}}}^{(1)} \triangleq \ma{F}^{(1)} + \Delta \ma{F}^{(1)}.
\end{equation}
This results in an expression for the perturbation $\Delta \ma{F}^{(1)}$ in the first factor matrix as
\begin{align}
\Delta \ma{F}^{(1)} =&  \Delta \ma{U}_{1}^{[\mathrm{s}]} \cdot \ma{T}_{1} +  \ma{U}_{1}^{[\mathrm{s}]} \cdot \Delta \ma{T}_{1} + \mathcal{O}(\Delta^2). \label{eq:SECSI_DeltaF1}
\end{align}
\textcolor{c22}{Using the result of \textcolor{c44}{the} JEVD defined in \cref{eq:RHS}}, \textcolor{c44}{a scaled version of the $k$th row of the factor matrix} $\ma{F}^{(3)}$  can be estimated from the diagonal of $\hat{\ma{D}}_{3,k}, \forall k=1,2,\dots, M_3,$ according to \cref{D3k} as
\begin{align*}
\hat{\tilde{\ma{F}}}^{(3)}(k,:) &=\mathrm{diag}\left( \hat{\ma{D}}_{3,k} \right)^{\mathrm{T}}\\
&= \mathrm{diag} \left(  \ma{D}_{3,k} + \Delta  \ma{D}_{3,k} \right)^{\mathrm{T}} \\
&= \tilde{\ma{F}}^{(3)}(k,:) + \mathrm{diag} \left(  \Delta \ma{D}_{3,k} \right)^{\mathrm{T}}  + \mathcal{O}(\Delta^2) \, ,
\end{align*}
\textcolor{c22}{where the perturbation in $\hat{\tilde{\ma{F}}}^{(3)}(k,:)$ is defined as
\begin{align*}
\hat{\tilde{\ma{F}}}^{(3)}(k,:) = \tilde{\ma{F}}^{(3)}(k,:) + \Delta \tilde{\ma{F}}^{(3)}(k,:),
\end{align*}
} which results in
\begin{align*}
\Delta \tilde{\ma{F}}^{(3)}(k,:) = \mathrm{diag} \left( \Delta \ma{D}_{3,k} \right)^{\mathrm{T}} + \mathcal{O}(\Delta^2).
\end{align*}
\textcolor{c22}{To get an expression \textcolor{c44}{of the corresponding factor matrix estimate} $\hat{{\ma{F}}}^{(3)}$, we take into account \cref{D3k} via  $\hat{\ma{F}}^{(3)} \triangleq \hat{\tilde{\ma{F}}}^{(3)}\cdot \mathrm{Diag} \left( \ma{F}^{(3)}(p,:) \right)$}. This leads to
%
\begin{equation}\label{eq:SECSI_F3}
\hat{{\ma{F}}}^{(3)} = \ma{F}^{(3)} + \Delta \ma{F}^{(3)},
\end{equation}
where
\begin{equation}\label{eq:SECSI_DeltaF3}
\Delta \ma{F}^{(3)} = \Delta \tilde{\ma{F}}^{(3)} \cdot \mathrm{Diag} \left( \ma{F}^{(3)}(p,:) \right).
\end{equation}
Using this equation, the perturbation \textcolor{c44}{of the $k$th row of} $\Delta \ma{F}^{(3)}$ is obtained as
\begin{equation}\label{eq:SECSI_DeltaF3k}
\Delta \ma{F}^{(3)}(k,:) =  \mathrm{diag} \left( \Delta \ma{D}_{3,k} \right)^{\mathrm{T}} \cdot \mathrm{Diag} \left( \ma{F}^{(3)}(p,:) \right) +\mathcal{O}(\Delta^2).
\end{equation}
\textcolor{c44}{The factor matrix} $\ma{F}^{(2)}$ can be estimated via a LS fit. \textcolor{c44}{To this end, we define the \textcolor{c77}{LS} estimate of $\ma{F}^{(2)}$ as
\begin{align*}
\hat{\ma{F}}^{(2)} \triangleq [{\ten{X}}]_{(2)} \cdot \left( \hat{{\ma{F}}}^{(3)} \diamond \hat{\ma{F}}^{(1)} \right) ^{+\mathrm{T}}.
\end{align*}
Using equations \eqref{eq:SECSI_F1} and \eqref{eq:SECSI_F3}, we get:}
\begin{equation*}
\hat{\ma{F}}^{(2)} = \Big[\ten{X}_0 + \ten{N}\Big] _{(2)} \cdot \left[ \left(\ma{F}^{(3)} + \Delta \ma{F}^{(3)}\right) \diamond \left(\ma{F}^{(1)} + \Delta \ma{F}^{(1)}\right) \right] ^{+\mathrm{T}}.
\end{equation*}
Since $ \hat{\ma{F}}^{(2)} =  \ma{F}^{(2)} + \Delta \ma{F}^{(2)}$, we finally calculate $\Delta \ma{F}^{(2)}$ as
\begin{align}
&\Delta \ma{F}^{(2)} = \hat{\ma{F}}^{(2)} -  \ma{F}^{(2)} \nonumber\\
&=
\textcolor{c7}{ \Big[ \ten{X}_0 + \ten{N}\Big] _{(2)} \left[ \left(\ma{F}^{(3)} + \Delta \ma{F}^{(3)} \right) \diamond \left(\ma{F}^{(1)} + \Delta \ma{F}^{(1)} \right) \right]^{+\mathrm{T}} } - \ma{F}^{(2)} \nonumber
\end{align}
\textcolor{c44}{Using the Taylor expansion, we get}
\begin{align}
&\Delta \ma{F}^{(2)}=
[\ten{X}_0]_{(2)} \cdot \Big[ \left(\ma{F}^{(3)} \diamond \ma{F}^{(1)} \right) + \left(\Delta\ma{F}^{(3)} \diamond \ma{F}^{(1)} \right) \nonumber\\
&\quad + \left(\ma{F}^{(3)} \diamond \Delta\ma{F}^{(1)} \right) \Big]^{+\mathrm{T}} + [\ten{N}]_{(2)} \cdot \left[ \ma{F}^{(3)} \diamond \ma{F}^{(1)} \right]^{+\mathrm{T}} \nonumber\\
 &\quad - \ma{F}^{(2)} \textcolor{c44}{+\mathcal{O}(\Delta^2)} \nonumber\\
&=
[\ten{X}_0]_{(2)} \cdot \Big[ \left(\ma{F}^{(3)} \diamond \ma{F}^{(1)} \right)^{+} - \left(\ma{F}^{(3)} \diamond \ma{F}^{(1)} \right)^{+} \nonumber\\
&\quad \cdot \Big(\big(\Delta\ma{F}^{(3)} \diamond \ma{F}^{(1)} \big)  + \big(\ma{F}^{(3)} \diamond \Delta\ma{F}^{(1)} \big)\Big) \cdot \left(\ma{F}^{(3)} \diamond \ma{F}^{(1)} \right)^{+} \Big]^{\mathrm{T}} \nonumber\\
&\quad + [\ten{N}]_{(2)} \cdot \left[ \ma{F}^{(3)} \diamond \ma{F}^{(1)} \right]^{+\mathrm{T}}  - \ma{F}^{(2)} + \mathcal{O}(\Delta^2) \nonumber\\
&=
-\ma{F}^{(2)} \cdot \left[ \left(\Delta\ma{F}^{(3)} \diamond \ma{F}^{(1)} \right) + \left(\ma{F}^{(3)} \diamond \Delta\ma{F}^{(1)} \right)\right]^{\mathrm{T}} \nonumber\\
&\quad \cdot \left(\ma{F}^{(3)} \diamond \ma{F}^{(1)} \right)^{+\mathrm{T}}  + [\ten{N}]_{(2)} \cdot \left[ \ma{F}^{(3)} \diamond \ma{F}^{(1)} \right]^{+\mathrm{T}} + \mathcal{O}(\Delta^2). \label{eq:SECSI_DeltaF2}
\end{align}
In the same manner, another set of \textcolor{c44}{factor matrix} estimates can be obtained by solving the lhs JEVD problem for $\hat{\ma{T}}_2$ and $\hat{\ma{D}}_{3,k}, \forall k=1,2\dots, M_3$ \textcolor{c44}{in \cref{eq:LHS}}. This leads to the two sets of estimates \textcolor{c44}{(rhs and lhs)} in the third mode, as summarized in Table \ref{alg:SECSI}.
%
%
%
Note that the obtained expressions can be directly used for the first and second mode estimates \textcolor{c44}{obtained from \cref{1_mode} and \cref{2_mode}, respectively}, since such estimates can be derived by applying the SECSI framework  to a permuted version of $\ten{X}$. For example, we can obtain the first mode estimates, \textcolor{c44}{corresponding to \cref{1_mode}}, by applying the SECSI framework on the third mode of $\mathrm{permute}(\ten{X},[2,3,1])$, where $\mathrm{permute}(\ten{A}, \text{ORDER})$ rearranges the dimensions of $\ten{A}$ so that they are in the order specified by the vector ORDER (as defined in Matlab). In the same manner, the second mode estimates, \textcolor{c44}{corresponding to \cref{2_mode}}, are obtained by using the SECSI framework on the third mode of $\mathrm{permute}(\ten{X},[1,3,2])$. \textcolor{c44}{Therefore, we obtain a total of six estimates for each factor matrix (two from each mode), as shown in Fig. \ref{overview_SECSI}. To select the final estimates, we can use best matching scheme or any low-complexity heuristic alternative that has been discussed in \cite{SECSI}.}
\section{Closed-Form rMSFE expressions}\label{sec:SECSI_Performance}
In this section, we present an analytical rMSFE expression \textcolor{c7}{for} each of the factor matrices \textcolor{c44}{estimates}. We first introduce some definitions which will be used subsequently to derive the analytical expressions. \textcolor{c88}{Additionally}, Theorem 1 will be used to resolve the scaling ambiguity in the factor matrices of the CP decomposition.
\subsection{Preliminary Definitions}\label{sec:SECSI_Preliminary}
Let $\ma{P}_{(M_1,M_2,M_3)}^{(r)} \in \{0,1\}^{(M_1\cdot M_2 \cdot M_3) \times (M_1\cdot M_2 \cdot M_3)}$ be the $r$-to-1 mode permutation matrix of any third order tensor $\ten{Z}\in \textcolor{c7}{\compl^{M_1 \times M_2 \times M_3}}$. This means that \textcolor{c44}{the permutation matrix $\ma{P}_{(M_1,M_2,M_3)}^{(r)}$ satisfies the property}
\begin{equation}\label{eq:P_r-to-1}
\mathrm{vec}\left\{ [\ten{Z}]_{(r)} \right\} = \ma{P}_{(M_1,M_2,M_3)}^{(r)} \cdot \mathrm{vec}\left\{ [\ten{Z}]_{(1)} \right\}.
\end{equation}
Note that $\ma{P}_{(M_1,M_2,M_3)}^{(1)} = \ma{I}_{(M_1\cdot M_2 \cdot M_3)}$. In the same manner, let $\ma{Q}_{(M_1,M_2)} \in \{0,1\}^{(M_1\cdot M_2) \times (M_1\cdot M_2)}$ be the permutation matrix \textcolor{c44}{that satisfies the following relation} for any $\ma{Z} \in \compl^{M_1 \times M_2}$
\begin{equation}\label{eq:Q_H-to-c}
\mathrm{vec}\left\{ \ma{Z}^{\mathrm{T}} \right\} = \ma{Q}_{(M_1,M_2)} \cdot \mathrm{vec}\left\{ \ma{Z} \right\} .
\end{equation}
\textcolor{c88}{Additionally}, let $\ma{n}_r \triangleq \mathrm{vec}\left\{ [\ten{N}]_{(r)} \right\}$ be the $r$-mode noise vector, with $\ma{R}_{\mathrm{nn}}^{(r)} \triangleq \mathbb{E}\{\ma{n}_r \cdot \ma{n}_r^{\mathrm{H}}\}$ and $\ma{C}_{\mathrm{nn}}^{(r)} \triangleq \mathbb{E}\{\ma{n}_r \cdot \textcolor{c7}{\ma{n}_r^{\mathrm{T}}}\}$ being the corresponding $r$-mode covariance and pseudo-covariance matrices\textcolor{c7}{,} respectively. Note that the these covariance matrices $\ma{R}_{\mathrm{nn}}^{(r)}$ are permuted versions of $\ma{R}_{\mathrm{nn}}^{(1)}$, since $\ma{n}_r = \ma{P}_{(M_1,M_2,M_3)}^{(r)} \cdot \ma{n}_1$. This property \textcolor{c44}{is also satisfied} for the pseudo-covariance matrices $\ma{C}_{\mathrm{nn}}^{(r)}$.

Let $\ma{W}_{(d)} \in \{0,1\}^{ d^2 \times d^2}$ be the diagonal elements selection matrix defined as
\begin{equation}\label{eq:SECSI_Wd}
\mathrm{vec}\left\{ \mathrm{Ddiag}\left(\ma{Z}\right) \right\} = \ma{W}_{(d)} \cdot  \mathrm{vec}\left\{ \ma{Z} \right\} \in \compl^{d^2 \times 1},
\end{equation}
where $\ma{Z} \in \compl^{d \times d}$ is a square matrix. \textcolor{c88}{Note that $\bm{W}_{(d)}$ is simply $\bm{I}_{d^2}-\bm{J}_{(d)}$, where $\bm{J}_{(d)}$ has already been defined below \cref{eq:w_secsi}}. Likewise, let $\ma{W}_{(d)}^{\mathrm{red}} \in \{0,1\}^{ d \times d^2}$ be the reduced \textcolor{c44}{dimensional} diagonal elements selection matrix \textcolor{c44}{that selects only the diagonal elements i.e.,}
\begin{equation}\label{eq:SECSI_Wred}
\mathrm{diag}\left(\ma{Z}\right) = \ma{W}_{(d)}^{\mathrm{red}} \cdot  \mathrm{vec}\left\{ \ma{Z} \right\} \in \compl^{d \times 1}
\end{equation}
for any square matrix $\ma{Z} \in \compl^{d \times d}$.
\begin{theorem}\label{teo:Popt}
	\textcolor{c88}{Let $\ma{Z}\in \compl^{d \times d}$ and $\Delta \ma{Z}\in \compl^{d \times d}$ be two matrices where the norm of each column in $\Delta \ma{Z}$ is much smaller than the corresponding column in $\ma{Z}$. Let  $\tilde{\ma{P}}\in \compl^{d \times d}$ be a diagonal matrix that introduces a scaling ambiguity in $(\ma{Z} + \Delta \ma{Z})$.} Then, a diagonal matrix $\ma{P}_{\mathrm{opt}}$ that resolve this scaling ambiguity can be expressed as	\textcolor{c77}{
	\begin{equation*}
	\ma{P}_{\mathrm{opt}} = \argmin_{\footnotesize{\ma{P}}\in \mathcal{M}_\text{D}(d)} \left\|\ma{Z} - (\ma{Z} + \Delta \ma{Z}) \cdot \tilde{\ma{P}} \cdot \ma{P}  \right\|_{\mathrm{F}}^2,
	\end{equation*}
where $\mathcal{M}_\text{D}(d)$ is the set of $d \times d$ diagonal matrices}. Thus, the following relation holds
	\begin{align*}
	&\ma{Z}-(\ma{Z} + \Delta \ma{Z}) \cdot \tilde{\ma{P}} \cdot \ma{P}_{\mathrm{opt}}    \\
	&=   \ma{Z} \cdot \mathrm{Ddiag}\left( \ma{Z}^{\mathrm{H}} \cdot  \Delta\ma{Z} \right)\cdot \bm{K}^{-1}- \Delta\ma{Z} +\mathcal{O}(\Delta^2),
	\end{align*}
where $\bm{K} = \mathrm{Ddiag}( \ma{Z}^{\mathrm{H}} \cdot{\ma{Z} })\in \real^{d \times d}$ is a diagonal matrix.
\end{theorem}
\begin{proof}
	cf. Appendix \ref{proof:Popt}
\end{proof}
\subsection{Factor Matrices rMSFE Expressions}
In this section, we derive closed-form $\mathrm{rMSFE}$ expressions for \textcolor{c44}{the three} factor matrices. \textcolor{c44}{As an example, we consider these closed-form $\mathrm{rMSFE}$ expressions using the JEVD of the rhs tensor slices in \cref{eq:RHS}. In Section \ref{sec:SECSI_LHS}, we also present the results for the lhs tensor slices in \cref{eq:LHS}.} The $\mathrm{rMSFE}$ in the $r$-th factor matrix is defined as
\begin{equation}\label{eq:SECSI_rMSFEdef}
\mathrm{rMSFE}^{(r)} = \mathbb{E}\left\{\min_{\footnotesize{\ma{P}}^{(r)}\in \mathcal{M}_\text{PD}(d)} \left( \frac{ \left\| \ma{F}^{(r)} -  \hat{\ma{F}}^{(r)}\cdot \ma{P}^{(r)} \right\|_{\mathrm{F}}^{2}}{ \left\| \ma{F}^{(r)}\right\|_{\mathrm{F}}^{2}} \right) \right\},
\end{equation}
\textcolor{c44}{where $\mathcal{M}_\text{PD}(d)$ is the set of $d \times d$ permuted diagonal matrices (also called monomial matrices), i.e., the matrices $\ma{P}^{(r)}$ correct the permutation and scaling ambiguity that is inherent in the estimation of the loading matrices \cite{SECSI} and \textcolor{c99}{$\ma{F}^{(r)}$ is the true factor matrix. The goal of the SECSI framework is to estimate these factor matrices up to the inevitable ambiguities, and our goal in the performance analysis is to predict the resulting relative errors, assuming (for the sake of the analysis) that the true matrices are known and that these ambiguities are resolved.}
Consequently, the rMSFE measures how accurately the actual CP model can be estimated from the noisy observations. }
%
\textcolor{c88}{After correcting the permutation ambiguity, the factor matrix estimates that we get from Monte-Carlo simulations \cite{SECSI} can be approximated as
\begin{align*}
 \hat{\ma{F}}^{(r)} &=({\ma{F}}^{(r)} + \Delta {\ma{F}}^{(r)})\cdot\tilde{\ma{P}}^{(r)},
\end{align*}
where $\Delta {\ma{F}}^{(r)}$ represents the perturbation that can be obtained from the performance analysis and $\tilde{\ma{P}}^{(r)}$ is a diagonal matrix modeling the scaling ambiguity. By using this relation, we rewrite the $\mathrm{rMSFE}$ definition in \cref{eq:SECSI_rMSFEdef} as}
\begin{equation}\label{eq:SECSI_rMSFEdef2}
\mathrm{rMSFE}^{(r)} = \mathbb{E}\left\{\frac{ \left\| \ma{F}^{(r)} -  (\ma{F}^{(r)} +  \Delta \ma{F}^{(r)})\cdot\tilde{\ma{P}}^{(r)}\cdot \ma{P}_{\mathrm{opt}}^{(r)} \right\|_{\mathrm{F}}^{2}}{ \left\| \ma{F}^{(r)}\right\|_{\mathrm{F}}^{2}} \right\},
\end{equation}
\textcolor{c88}{where  $\ma{P}_{\mathrm{opt}}^{(r)}$ is the optimal column scaling matrix, since the scaling ambiguity is only relevant for the perturbation analysis.} \textcolor{c44}{To derive closed-form $\mathrm{rMSFE}^{(r)}$ expressions, we first vectorize $ \ma{F}^{(r)} -  (\ma{F}^{(r)} +  \Delta \ma{F}^{(r)})\cdot\tilde{\ma{P}}^{(r)}\cdot \ma{P}_{\mathrm{opt}}^{(r)}$, use Theorem \ref{teo:Popt}, and the definitions in \cref{eq:SECSI_Wd} and \cref{eq:Q_H-to-c}, to get}
\textcolor{c88}{
\begin{align}
&\mathrm{vec}\left\{ \ma{F}^{(r)} -  (\ma{F}^{(r)} +  \Delta \ma{F}^{(r)}) \cdot\tilde{\ma{P}}^{(r)}\cdot\ma{P}_{\mathrm{opt}}^{(r)} \right\} \nonumber\\
&\approx
\mathrm{vec}\bigg\{\ma{F}^{(r)} \cdot \mathrm{Ddiag}\left( \ma{F}^{(r)^{\mathrm{H}}} \cdot  \Delta\ma{F}^{(r)} \right) \cdot \bm{K}_r^{-1}  -  \Delta\ma{F}^{(r)} \bigg\}\nonumber\\
&= \left( \ma{I}_{d} \otimes \ma{F}^{(r)} \right) \cdot \left( \bm{K}_r^{-1} \otimes \ma{I}_{d} \right)\cdot \nonumber\\
 &\quad\ma{W}_{(d)} \cdot \mathrm{vec}\Big\{ \ma{F}^{(r)^{\mathrm{H}}}
 \cdot  \Delta\ma{F}^{(r)} \Big\} - \mathrm{vec}\left\{ \Delta \ma{F}^{(r)}  \right\}\nonumber\\
&= \left( \ma{I}_{d} \otimes \ma{F}^{(r)} \right) \cdot \left( \bm{K}_r^{-1} \otimes \ma{I}_{d} \right) \cdot \ma{W}_{(d)}\cdot \nonumber\\
&\quad    \left( \ma{I}_{d} \otimes \ma{F}^{(r)^{\mathrm{H}}} \right)\cdot  \mathrm{vec}\left\{ \Delta \ma{F}^{(r)}  \right\}
 - \mathrm{vec}\left\{ \Delta \ma{F}^{(r)}  \right\}, \label{eq:SECSI_vec_rMSFEAll}
\end{align}
where $\bm{K}_r = \mathrm{Ddiag}\left( \ma{F}^{(r)^{\mathrm{H}}} \cdot  \ma{F}^{(r)} \right)$.}
Note that the resulting expression contains the vectorization of the perturbation in the respective factor matrix estimates.
As shown in Section \ref{Pert_factor_matrix}, the perturbations in the three factor matrix estimates (\cref{eq:SECSI_DeltaF1}, \cref{eq:SECSI_DeltaF2}, and \cref{eq:SECSI_DeltaF3}) are a function of different perturbations, i.e., $\Delta \ma{U}_{r}^{[\mathrm{s}]}$ in \cref{eq:SECSI_Delta_Ur}, $\Delta \ten{S}^{[\mathrm{s}]}$ in \cref{eq:SECSI_DeltaSs}, $\Delta \ten{S}_3^{[\mathrm{s}]}$ in \cref{eq:SECSI:Delta_S3}, $\Delta \ma{S}_{3,k}$ in \cref{eq:Delta_S3k}, $\Delta{\ma{S}}_{3,k}^{\mathrm{rhs}}$ in \cref{eq:DeltaS3_RHS}, and $\Delta \ma{D}_{3,k}$ in \cref{eq:SECSI_vec_Delta_D3k}. Therefore, to get final closed-form $\mathrm{rMSFE}$ expressions for the \textcolor{c7}{three} factor matrices in \cref{eq:SECSI_vec_rMSFEAll}, we \textcolor{c77}{vectorize} all of the perturbation expressions.
We start by applying the vectorization operator \textcolor{c44}{to} $\Delta \ma{U}_{r}^{[\mathrm{s}]}$ \textcolor{c44}{in}  \cref{eq:SECSI_Delta_Ur} and \textcolor{c44}{use} the $r$-to-1 mode permutation matrix \textcolor{c44}{in} \cref{eq:P_r-to-1} \textcolor{c44}{to get}
\begin{align}
&\mathrm{vec}\left\{ \Delta \ma{U}_{r}^{[\mathrm{s}]} \right\}
= \left(\ma{\Sigma}_r^{[\mathrm{s}]^{-1}} \ma{V}_r^{[\mathrm{s}]^{\mathrm{T}}} \otimes \ma{\Gamma}_r^{[\mathrm{n}]}\right) \cdot \ma{n}_r  + \mathcal{O}(\Delta^2)\nonumber\\
&= \left(\ma{\Sigma}_r^{[\mathrm{s}]^{-1}} \ma{V}_r^{[\mathrm{s}]^{\mathrm{T}}} \otimes \ma{\Gamma}_r^{[\mathrm{n}]}\right) \cdot \ma{P}_{(M_1,M_2,M_3)}^{(r)} \cdot \ma{n}_1 + \mathcal{O}(\Delta^2). \label{eq:SECSI_vec_Delta_Ur}
\end{align}
\textcolor{c44}{Next}, we \textcolor{c77}{vectorize the 1-mode} unfolding of $\Delta \ten{S}^{[\mathrm{s}]}$\textcolor{c7}{ in \cref{eq:SECSI_DeltaSs}}, by using the definition in \cref{eq:P_r-to-1}  to obtain\textcolor{c111}{
\begin{align}
\mathrm{vec}&\left\{[\Delta \ten{S}^{[\mathrm{s}]}]_{(1)}\right\} =
\mathrm{vec}\left\{\left[\ten{N} \overset{3}{\underset{r=1}{ {\bigtimes}_r } } \ma{U}_{r}^{[\mathrm{s}]^{\mathrm{H}}} \right]_{(1)}\right\} + \mathcal{O}(\Delta^2) \nonumber\\
&= \underbrace{\left( \ma{U}_{2}^{[\mathrm{s}]^{\mathrm{H}}} \otimes \ma{U}_{3}^{[\mathrm{s}]^{\mathrm{H}}}  \otimes \ma{U}_{1}^{[\mathrm{s}]^{\mathrm{H}}}\right)}_{\ma{L}_{0}}\cdot \ma{n}_1  + \mathcal{O}(\Delta^2). \label{eq:SECSI_vec_Delta_S}
\end{align}}

In the same manner, we expand $\mathrm{vec}\left\{[\Delta \ten{S}_{3}^{[\mathrm{s}]}]_{(1)}\right\}$ using  \cref{eq:SECSI_vec_Delta_Ur} \textcolor{c7}{and} \cref{eq:SECSI_vec_Delta_S} to get

\small
\begin{align}
&\mathrm{vec}\left\{[\Delta \ten{S}_{3}^{[\mathrm{s}]}]_{(1)}\right\} \nonumber \\
&=
\mathrm{vec}\left\{ [\Delta \ten{S}^{[\mathrm{s}]} \times _{3} \ma{U}_{3}^{[\mathrm{s}]} + \ten{S}^{[\mathrm{s}]} \times _{3} \Delta \ma{U}_{3}^{[\mathrm{s}]}]_{(1)}\right\}  + \mathcal{O}(\Delta^2) \nonumber\\
&= \ma{P}_{(d,d,M_3)}^{(3)^{\mathrm{T}}} \cdot \mathrm{vec}\left\{ [\Delta \ten{S}^{[\mathrm{s}]} \times _{3} \ma{U}_{3}^{[\mathrm{s}]} + \ten{S}^{[\mathrm{s}]} \times _{3} \Delta \ma{U}_{3}^{[\mathrm{s}]}]_{(3)}\right\}  \nonumber\\
&+ \mathcal{O}(\Delta^2)  \nonumber\\
&= \ma{P}_{(d,d,M_3)}^{(3)^{\mathrm{T}}} \cdot \left(\ma{I}_{d^2} \otimes \ma{U}_{3}^{[\mathrm{s}]}\right) \cdot \mathrm{vec}\left\{[\Delta \ten{S}^{[\mathrm{s}]}]_{(3)}\right\} \nonumber\\
&\, + \ma{P}_{(d,d,M_3)}^{(3)^{\mathrm{T}}} \cdot \left([ \ten{S}^{[\mathrm{s}]}]_{(3)}^{\mathrm{T}} \otimes \ma{I}_{M_3}\right) \cdot \mathrm{vec}\left\{\Delta \ma{U}_{3}^{[\mathrm{s}]}\right\} + \mathcal{O}(\Delta^2) \nonumber\\
&=
\ma{P}_{(d,d,M_3)}^{(3)^{\mathrm{T}}} \cdot \left(\ma{I}_{d^2} \otimes \ma{U}_{3}^{[\mathrm{s}]}\right) \cdot \ma{P}_{(d,d,d)}^{(3)}\cdot  \mathrm{vec}\left\{[\Delta \ten{S}^{[\mathrm{s}]}]_{(1)}\right\} \nonumber\\
&\, + \ma{P}_{(d,d,M_3)}^{(3)^{\mathrm{T}}} \cdot \left([\Delta \ten{S}^{[\mathrm{s}]}]_{(3)}^{\mathrm{T}} \otimes \ma{I}_{M_3}\right) \cdot \left(\textcolor{c7}{\ma{\Sigma}_3^{[\mathrm{s}]^{-1}}} \ma{V}_3^{[\mathrm{s}]^{\mathrm{T}}} \otimes \ma{\Gamma}_3^{[\mathrm{n}]}\right) \nonumber\\
&\, \cdot \ma{P}_{(M_1,M_2,M_3)}^{(3)} \cdot \ma{n}_1 + \mathcal{O}(\Delta^2) \nonumber\\
&=
\ma{P}_{(d,d,M_3)}^{(3)^{\mathrm{T}}} \cdot \left(\ma{I}_{d^2} \otimes \ma{U}_{3}^{[\mathrm{s}]}\right) \cdot \ma{P}_{(d,d,d)}^{(3)}\cdot   \ma{L}_{0} \cdot \ma{n}_1  + \ma{P}_{(d,d,M_3)}^{(3)^{\mathrm{T}}} \cdot \nonumber\\
&\,\left([\Delta \ten{S}^{[\mathrm{s}]}]_{(3)}^{\mathrm{T}}\cdot \textcolor{c7}{\ma{\Sigma}_3^{[\mathrm{s}]^{-1}}} \ma{V}_3^{[\mathrm{s}]^{\mathrm{T}}}  \otimes \ma{\Gamma}_3^{[\mathrm{n}]}\right) \cdot \ma{P}_{(M_1,M_2,M_3)}^{(3)}    \cdot \ma{n}_1 + \mathcal{O}(\Delta^2) \nonumber\\
& = \ma{L}_{1} \cdot \ma{n}_1 + \mathcal{O}(\Delta^2) \, , \label{eq:SECSI_vec_Delta_S3}
\end{align}
\normalsize
where
\begin{align}
&\ma{L}_{1} \triangleq
\ma{P}_{(d,d,M_3)}^{(3)^{\mathrm{T}}} \cdot \left(\ma{I}_{d^2} \otimes \ma{U}_{3}^{[\mathrm{s}]}\right) \cdot \ma{P}_{(d,d,d)}^{(3)} \cdot \ma{L}_{0}
+ \ma{P}_{(d,d,M_3)}^{(3)^{\mathrm{T}}} \cdot \nonumber\\&\left([ \ten{S}^{[\mathrm{s}]}]_{(3)}^{\mathrm{T}}\cdot \textcolor{c7}{\ma{\Sigma}_3^{[\mathrm{s}]^{-1}}} \ma{V}_3^{[\mathrm{s}]^{\mathrm{T}}}  \otimes \ma{\Gamma}_3^{[\mathrm{n}]}\right) \cdot \ma{P}_{(M_1,M_2,M_3)}^{(3)}
\label{eq:SECSI_L1}
\end{align}
%
Furthermore, we can use this result to obtain an expression for the vectorization of each slice $\Delta \ma{S}_{3,k} \quad \forall k=1,2,\dots, M_3$ of $\Delta \ten{S}_3^{[\mathrm{s}]}$. Using \textcolor{c44}{\cref{eq:Delta_S3k}} and \cref{eq:SECSI_vec_Delta_S3}, we derive an expression for $\mathrm{vec}\left\{\Delta \ma{S}_{3,k}\right\}$ as
\begin{align}
\mathrm{vec}\left\{\Delta \ma{S}_{3,k}\right\} &= \mathrm{vec}\left\{\Delta \ten{S}_3^{[\mathrm{s}]} \times _{3} \ma{e}_{M_3, k}^{\mathrm{T}}\right\} \nonumber\\
&= \mathrm{vec}\left\{[\Delta \ten{S}_3^{[\mathrm{s}]} \times _{3} \ma{e}_{M_3, k}^{\mathrm{T}}]_{(1)}\right\} \nonumber \\
&= \mathrm{vec}\left\{[\Delta \ten{S}_3^{[\mathrm{s}]}]_{(1)} \cdot (\ma{I}_{d} \otimes \ma{e}_{M_3, k}^{\mathrm{T}})^{\mathrm{T}}\right\}  + \mathcal{O}(\Delta^2) \nonumber\\
&= (\ma{I}_{d} \otimes \ma{e}_{M_3, k}^{\mathrm{T}} \otimes \ma{I}_{d})  \cdot \mathrm{vec}\left\{[\Delta \ten{S}_3^{[\mathrm{s}]}]_{(1)}\right\}  \nonumber\\
&= (\ma{I}_{d} \otimes \ma{e}_{M_3, k}^{\mathrm{T}} \otimes \ma{I}_{d})  \cdot  \ma{L}_{1} \cdot \ma{n}_1  \nonumber\\
&=  \ma{L}_{2}^{(k)} \cdot \ma{n}_1 + \mathcal{O}(\Delta^2) \, , \label{eq:SECSI_vec_Delta_S3k}
\end{align}
where
\begin{align}
\ma{L}_{2}^{(k)} &\triangleq
(\ma{I}_{d} \otimes \ma{e}_{M_3, k}^{\mathrm{T}} \otimes \ma{I}_{d})  \cdot \ma{L}_{1} \label{eq:SECSI_L2k}
\end{align}
for $k=1,2,\dots,M_3$. \textcolor{c44}{Next}, we use this result \textcolor{c7}{and \cref{eq:SECSI_DeltaSs}} to expand the vectorization of  \cref{eq:DeltaS3_RHS}. This leads to
\begin{align}
&\mathrm{vec}\left\{\Delta \ma{S}_{3,k}^{\mathrm{rhs}}\right\} = \nonumber \\
&= \mathrm{vec}\left\{ \Delta \ma{S}_{3,k} \cdot \ma{S}_{3,p}^{-1} - \ma{S}_{3,k} \cdot \ma{S}_{3,p}^{-1} \cdot \Delta \ma{S}_{3,p} \cdot \ma{S}_{3,p}^{-1} \right\} + \mathcal{O}(\Delta^2)  \nonumber\\
&= ( \ma{S}_{3,p}^{-\mathrm{T}} \otimes \ma{I}_{d}) \cdot \mathrm{vec}\left\{\Delta \ma{S}_{3,k}\right\} - ( \ma{S}_{3,p}^{-\mathrm{T}} \otimes \ma{S}_{3,k} \cdot \ma{S}_{3,p}^{-1}) \nonumber\\
&\quad \cdot \mathrm{vec}\left\{\Delta \ma{S}_{3,p}\right\}  + \mathcal{O}(\Delta^2) \nonumber\\
&= ( \ma{S}_{3,p}^{-\mathrm{T}} \otimes \ma{I}_{d}) \cdot \ma{L}_{2}^{(k)} \cdot \ma{n}_1  - ( \ma{S}_{3,p}^{-\mathrm{T}} \otimes \ma{S}_{3,k} \cdot \ma{S}_{3,p}^{-1}) \cdot  \ma{L}_{2}^{(p)} \cdot \ma{n}_1 \nonumber\\
&\quad + \mathcal{O}(\Delta^2) \nonumber\\
&=  \ma{L}_{3}^{(k)} \cdot \ma{n}_1 + \mathcal{O}(\Delta^2) \, , \label{eq:SECSI_vec_Delta_S3kRHS}
\end{align}
where
\begin{align}
\ma{L}_{3}^{(k)} &\triangleq
( \ma{S}_{3,p}^{-\mathrm{T}} \otimes \ma{I}_{d}) \cdot  \ma{L}_{2}^{(k)} - ( \ma{S}_{3,p}^{-\mathrm{T}} \otimes \ma{S}_{3,k} \cdot \ma{S}_{3,p}^{-1}) \cdot \ma{L}_{2}^{(p)}. \label{eq:SECSI_L3k}
\end{align}
Next, we stack the column vectors $\bm{s}_k=\mathrm{vec}\left\{\Delta \ma{S}_{3,k}^{\mathrm{rhs}}\right\}$ into the vector \textcolor{c77}{$\bm{s}$}, \textcolor{c44}{as defined in \cref{eq:BnAw_SECSI}}, and use the previous result to obtain
\begin{align}
\ma{s}
=  \ma{L}_{3} \cdot \ma{n}_{1} + \mathcal{O}(\Delta^2) , \label{eq:SECSI_s}
\end{align}
where  $\ma{L}_{3} = \left[ \ma{L}_{3}^{(1)} \quad \ma{L}_{3}^{(2)} \hdots \ma{L}_{3}^{(M_3)}\right] ^T$.
%
%
This expression for $\ma{s}$ is used to expand  \cref{eq:SECSI_vec_Delta_T1} into
\begin{align}
\mathrm{vec}\left\{\Delta\ma{T}_{1}\right\} &= -\ma{A}^{+} \cdot \ma{B} \cdot \ma{s} \, \textcolor{c7}{ + \mathcal{O}(\Delta^2)} \nonumber\\
&= \underbrace{-\ma{A}^{+} \cdot \ma{B} \cdot \ma{L}_{3}}_{\ma{L}_{4}} \cdot \ma{n}_1 { + \mathcal{O}(\Delta^2)}. \label{eq:SECSI_vec_DeltaT}
\end{align}
Since $\mathrm{diag} \left( \mathrm{Ddiag} \left( \ma{Z} \right) \right) = \mathrm{diag} \left(  \ma{Z}  \right)$ for any matrix $\ma{Z}$, we use the matrix $\ma{W}_{(d)}^{\mathrm{red}}$ (defined in \cref{eq:SECSI_Wred}) and \cref{eq:SECSI_vec_Delta_D3k} to vectorize \cref{eq:SECSI_DeltaF3k} as
\small
\begin{align}
&\mathrm{vec}\big\{\Delta \ma{F}^{(3)}(k,:)\big\} \nonumber \\
&=
\mathrm{vec}\left\{\mathrm{diag} \left( \Delta \ma{D}_{3,k} \right)^{\mathrm{T}} \cdot \mathrm{Diag} \left( \ma{F}^{(3)}(p,:) \right) \right\} \nonumber\\
&=
\mathrm{vec}\left\{\mathrm{diag} \left( \mathrm{Ddiag} \left( \ma{T}_{1}^{-1} \cdot \Delta \ma{S}_{3,k}^{\mathrm{rhs}} \cdot \ma{T}_{1} \right) \right)^{\mathrm{T}}  \mathrm{Diag} \left( \ma{F}^{(3)}(p,:) \right) \right\} \nonumber\\
&\quad + \mathcal{O}(\Delta^2) \nonumber\\
&=
\mathrm{vec}\left\{\mathrm{diag} \left( \ma{T}_{1}^{-1} \cdot \Delta \ma{S}_{3,k}^{\mathrm{rhs}} \cdot \ma{T}_{1} \right)^{\mathrm{T}} \cdot \mathrm{Diag} \left( \ma{F}^{(3)}(p,:) \right) \right\}  + \mathcal{O}(\Delta^2) \nonumber\\
&=
\mathrm{vec}\left\{\mathrm{diag} \left[ \ma{T}_{1}^{-1} \cdot \Delta \ma{S}_{3,k}^{\mathrm{rhs}} \cdot \ma{T}_{1}  \cdot \mathrm{Diag} \left( \ma{F}^{(3)}(p,:) \right) \right] \right\}
+ \mathcal{O}(\Delta^2) \nonumber\\
&=
\ma{W}_{(d)}^{\mathrm{red}} \cdot \mathrm{vec}\left\{ \ma{T}_{1}^{-1} \cdot \Delta \ma{S}_{3,k}^{\mathrm{rhs}} \cdot \ma{T}_{1} \cdot \mathrm{Diag} \left( \ma{F}^{(3)}(p,:) \right)\right\} + \mathcal{O}(\Delta^2) \nonumber\\
&=
\ma{W}_{(d)}^{\mathrm{red}} \cdot \left( \mathrm{diag} \left( \ma{F}^{(3)}(p,:) \right) \cdot \ma{T}_{1}^{\mathrm{T}} \otimes \ma{T}_{1}^{-1}  \right)\cdot \mathrm{vec}\left\{  \Delta \ma{S}_{3,k}^{\mathrm{rhs}}  \right\} \nonumber\\ &\quad+ \mathcal{O}(\Delta^2) \nonumber\\
&=
\ma{W}_{(d)}^{\mathrm{red}} \left( \mathrm{diag} \left( \ma{F}^{(3)}(p,:) \right) \cdot \ma{T}_{1}^{\mathrm{T}} \otimes \ma{T}_{1}^{-1}  \right) \cdot \ma{L}_{3}^{(k)} \cdot \ma{n}_{1}+ \mathcal{O}(\Delta^2) \nonumber\\
&=  \ma{L}_{4}^{(k)} \cdot \ma{n}_1 + \mathcal{O}(\Delta^2) \, , \label{K4_general}
\end{align}
\normalsize
where
\begin{align}
\ma{L}_{4}^{(k)} &\triangleq
\ma{W}_{(d)}^{\mathrm{red}} \cdot   \left(\mathrm{diag} \left( \ma{F}^{(3)}(p,:) \right) \cdot  \ma{T}_{1}^{\mathrm{T}} \otimes \ma{T}_{1}^{-1}  \right) \cdot  \ma{L}_{3}^{(k)} \, .\label{eq:SECSI_L4}
\end{align}
\textcolor{c44} {The expression for each row of $\Delta \ma{F}^{(3)}$ \textcolor{c55}{in \cref{K4_general}} can be used to obtain an expression for the vectorization of $\Delta \ma{F}^{(3)}$. To this end, we use the permutation matrix $\ma{Q}_{(M_3,d)}$, defined in \cref{eq:Q_H-to-c}, to express $\mathrm{vec}\left\{\Delta \ma{F}^{(3)} \right\}$ as}\textcolor{c88} {
%
%
\begin{align}
\mathrm{vec}\big\{\Delta \ma{F}^{(3)} \big\} &= \ma{Q}_{(M_3,d)}^{\mathrm{T}} \cdot \mathrm{vec}\left\{\Delta \ma{F}^{(3)^{\mathrm{T}}} \right\} \nonumber \\
&=
\ma{Q}_{(M_3,d)}^{\mathrm{T}} \cdot
\mathrm{stack}\{\mathrm{vec}\{\Delta \ma{F}^{(3)}(k,:) \} \} \nonumber \\
&= \ma{L}_{5} \cdot \ma{n}_{1} + \mathcal{O}(\Delta^2) \, , \label{eq:SECSI_vec_DeltaF3}
\end{align}
where $k=1,2,\cdots, M_3$ }and
\begin{align}
\ma{L}_{5}
= \ma{Q}_{(M_3,d)}^{\mathrm{T}} \cdot
\mathrm{stack}\{\ma{L}_{4}^{(k)}\} \label{eq:SECSI_KL5}
\end{align}
\textcolor{c77}{Then} we insert \cref{eq:SECSI_vec_DeltaF3} in \cref{eq:SECSI_vec_rMSFEAll} to get
\begin{align}
\mathrm{vec}\left\{ \ma{F}^{(3)} - (\ma{F}^{(3)} +  \Delta \ma{F}^{(3)})\cdot\tilde{\ma{P}}^{(3)}\cdot\ma{P}_{\mathrm{opt}}^{(3)} \right\} \approx \ma{L}_{\bm{F}_3} \cdot \ma{n}_1, \label{eq:SECSI_vec_rMSFE3}
\end{align}
where
\textcolor{c88}{
\begin{align}
\ma{L}_{\bm{F}_3} &=
\left( \ma{I}_{d} \otimes \ma{F}^{(3)} \right) \cdot \left( \bm{K}_3^{-1} \otimes \ma{I}_{d} \right) \cdot \ma{W}_{(d)}\cdot \nonumber\\
&\quad \Big[\left( \ma{I}_{d} \otimes \ma{F}^{(1)^{\mathrm{H}}} \right)\cdot \ma{L}_{5} \Big] - \ma{L}_{5} \, . \label{eq:SECSI_Lc}
\end{align}
The} final closed-form $\mathrm{rMSFE}_{\bm{F}_3}$ expression can be approximated to
\begin{align}
\mathrm{rMSFE}_{\bm{F}_3} &= \frac{ \mathrm{tr}\left(\ma{L}_{\bm{F}_3} \cdot  \ma{R}_{\mathrm{nn}}^{(1)} \cdot \ma{L}^{\mathrm{H}}_{\bm{F}_3} \right)}{ \left\| \ma{F}^{(3)}\right\|_{\mathrm{F}}^{2}} \, . \label{eq:SECSI_rMSFE3}
\end{align}
Similarly,  vectorization of the perturbation in the first factor matrix estimate $\Delta \ma{F}^{(1)}$ is performed by using \cref{eq:SECSI_DeltaF1}, \cref{eq:SECSI_vec_DeltaT}, and \cref{eq:SECSI_vec_Delta_Ur} as
\begin{align}
&\mathrm{vec}\big\{\Delta \ma{F}^{(1)} \big\} = \mathrm{vec}\left\{ \Delta \ma{U}_{1}^{[\mathrm{s}]} \cdot \ma{T}_{1} +  \ma{U}_{1}^{[\mathrm{s}]} \cdot \Delta \ma{T}_{1} \right\}  + \mathcal{O}(\Delta^2) \nonumber\\
&= \left( \ma{I}_{d} \otimes \ma{U}_{1}^{[\mathrm{s}]} \right) \cdot \mathrm{vec}\left\{\Delta\ma{T}_{1}\right\} + \left( \ma{T}_{1}^{\mathrm{T}} \otimes \ma{I}_{M_1} \right) \cdot \mathrm{vec}\left\{\Delta \ma{U}_{1}^{[\mathrm{s}]} \right\} \nonumber\\
&\quad + \mathcal{O}(\Delta^2) \nonumber \\
&= \left( \ma{I}_{d} \otimes \ma{U}_{1}^{[\mathrm{s}]} \right) \cdot \ma{L}_{4} \cdot \ma{n}_1  + \left( \ma{T}_{1}^{\mathrm{T}} \otimes \ma{I}_{M_1} \right) \nonumber\\
&\quad \cdot \left(\ma{\Sigma}_1^{[\mathrm{s}]^{-1}} \ma{V}_1^{[\mathrm{s}]^{\mathrm{T}}} \otimes \ma{\Gamma}_1^{[\mathrm{n}]}\right) \cdot \ma{n}_1  + \mathcal{O}(\Delta^2) \nonumber\\
&=  \ma{L}_{6} \cdot \ma{n}_1 + \mathcal{O}(\Delta^2) \, , \label{eq:SECSI_vec_DeltaF1}
\end{align}
where
\begin{align}
\ma{L}_{6} &= \left( \ma{I}_{d} \otimes \ma{U}_{1}^{[\mathrm{s}]} \right) \cdot \ma{L}_{4} +  \left( \ma{T}_{1}^{\mathrm{T}} \cdot \ma{\Sigma}_1^{[\mathrm{s}]^{-1}} \ma{V}_1^{[\mathrm{s}]^{\mathrm{T}}} \otimes \ma{\Gamma}_1^{[\mathrm{n}]} \right) \, . \label{eq:SECSI_L6}
\end{align}
By inserting \cref{eq:SECSI_vec_DeltaF1} in \cref{eq:SECSI_vec_rMSFEAll}, we get
\begin{align}
&\mathrm{vec}\left\{ \ma{F}^{(1)} - (\ma{F}^{(1)} +  \Delta \ma{F}^{(1)})\cdot\tilde{\ma{P}}^{(1)}\cdot\ma{P}_{\mathrm{opt}}^{(1)} \right\} \approx  \ma{L}_{\bm{F}_1} \cdot \ma{n}_1, \label{eq:SECSI_vec_rMSFE1}
\end{align}
where
\begin{align}
\ma{L}_{\bm{F}_1} &=
 \left( \ma{I}_{d} \otimes \ma{F}^{(1)} \right) \cdot \left( \bm{K}_1^{-1} \otimes \ma{I}_{d} \right) \cdot \ma{W}_{(d)}\cdot \nonumber\\
 &\quad \Big[ \left( \ma{I}_{d} \otimes \ma{F}^{(1)^{\mathrm{H}}} \right)\cdot \ma{L}_{6} \Big] - \ma{L}_{6} \, . \label{eq:SECSI_La}
\end{align}
Now we get the closed-form expression for the first factor matrix $\mathrm{rMSFE}_{\bm{F}_1}$ by
\begin{align}
\mathrm{rMSFE}_{\bm{F}_1} &=\frac{ \mathrm{tr}\left(\ma{L}_{\bm{F}_1} \cdot  \ma{R}_{\mathrm{nn}}^{(1)} \cdot \ma{L}^{\mathrm{H}}_{\bm{F}_1} \right)}{ \left\| \ma{F}^{(1)}\right\|_{\mathrm{F}}^{2}} \, . \label{eq:SECSI_rMSFE1}
\end{align}

Finally, \textcolor{c44}{the} vectorization of the perturbation in the second factor matrix estimate $\Delta \ma{F}^{(2)}$ is achieved by using the obtained expressions for the vectorization of $\Delta \ma{F}^{(1)}$ (\cref{eq:SECSI_vec_DeltaF1}) and $\Delta \ma{F}^{(3)}$ (\cref{eq:SECSI_vec_DeltaF3}) to vectorize $\Delta \ma{F}^{(2)}$ as
\begin{align}
&\mathrm{vec}\left\{ \Delta \ma{F}^{(2)} \right\}  \label{eq:vec_DeltaB}\\
&=
\mathrm{vec}\Big\{[\ten{N}]_{(2)} \cdot \left[ \ma{F}^{(3)} \diamond \ma{F}^{(1)} \right]^{+\mathrm{T}}  -\ma{F}^{(2)} \cdot \Big[ \left(\Delta\ma{F}^{(3)} \diamond \ma{F}^{(1)} \right) \nonumber\\ &  + \left(\ma{F}^{(3)} \diamond \Delta\ma{F}^{(1)} \right)\Big]^{\mathrm{T}} \cdot \left(\ma{F}^{(3)} \diamond \ma{F}^{(1)} \right)^{+\mathrm{T}} \Big\}  + \mathcal{O}(\Delta^2) \nonumber\\
&=
\left( \left(\ma{F}^{(3)} \diamond \ma{F}^{(1)} \right)^{+} \otimes \ma{I}_{M_2} \right) \cdot \ma{P}_{(M_1,M_2,M_3)}^{(2)} \cdot \ma{n}_1
 \nonumber\\ &- \left( \left(\ma{F}^{(3)} \diamond \ma{F}^{(1)} \right)^{+} \otimes \ma{F}^{(2)} \right) \cdot \ma{Q}_{(M_1\cdot M_2, d)}    \nonumber\\ &\quad \cdot \left[ \mathrm{vec}\left\{ \Delta\ma{F}^{(3)} \diamond \ma{F}^{(1)} \right\}  +  \mathrm{vec}\left\{ \ma{F}^{(3)} \diamond \Delta\ma{F}^{(1)} \right\}  \right] + \mathcal{O}(\Delta^2). \nonumber
\end{align}
%
\textcolor{c88}{The following relation for two matrices $\ma{X} = [\ma{x}_1, \ma{x}_2, \dots, \ma{x}_d] \in \compl^{M_1 \times d}$ and $\ma{Y} = [\ma{y}_1, \ma{y}_2, \dots, \ma{y}_d] \in \compl^{M_2 \times d}$ can easily be verified for the vectorization of Khatri-Rao products.
\begin{align}
\mathrm{vec}\left\{ \ma{X} \diamond \ma{Y} \right\} = \ma{G}(\ma{X}, M_2 ) \cdot \mathrm{vec}\{\ma{Y}\}
	= \ma{H}(\ma{Y}, M_1 ) \cdot \mathrm{vec}\{\ma{X}\}, \label{teo:KRP}
\end{align}
where
	\begin{align*}
	\ma{G}(\ma{X}, M_2 ) &\triangleq
	\begin{bmatrix}
	\left(\ma{x}_1 \otimes \ma{I}_{M_2}\right)\cdot \left(\ma{e}_{d, 1}^{\mathrm{T}} \otimes \ma{I}_{M_2}\right) \\
	\vdots \\
	\left(\ma{x}_d \otimes \ma{I}_{M_2}\right)\cdot \left(\ma{e}_{d, d}^{\mathrm{T}} \otimes \ma{I}_{M_2}\right) \\
	\end{bmatrix}
	\\
	\ma{H}(\ma{Y}, M_1 ) &\triangleq
	\begin{bmatrix}
	\left(\ma{I}_{M_1} \otimes \ma{y}_{1}\right)\cdot \left(\ma{e}_{d, 1}^{\mathrm{T}} \otimes \ma{I}_{M_1}\right) \\
	\vdots \\
	\left(\ma{I}_{M_1} \otimes \ma{y}_{d}\right)\cdot \left(\ma{e}_{d, d}^{\mathrm{T}} \otimes \ma{I}_{M_1}\right) \\
	\end{bmatrix}.
	\end{align*}
In order to apply the relation in \cref{teo:KRP}} to  \cref{eq:vec_DeltaB}, we define $\ma{f}_{l}^{(r)}$ to be the $l$-th column of the $r$-th factor matrix $\ma{F}^{(r)}$ which implies that $\ma{f}_{l}^{(r)} = \ma{F}^{(r)}(:,l)$. Therefore, by applying the relation in \cref{teo:KRP} to $\mathrm{vec}\left\{ \Delta\ma{F}^{(3)} \diamond \ma{F}^{(1)} \right\}$ we get
\begin{align}
\mathrm{vec}\left\{ \Delta \ma{F}^{(3)} \diamond \ma{F}^{(1)} \right\}
&= \ma{H}(\ma{F}^{(1)}, M_3 ) \cdot \mathrm{vec}\left\{ \Delta\ma{F}^{(3)} \right\}. \label{eq:KRP_1}
\end{align}
In the same manner, we apply \textcolor{c88}{the relation in \cref{teo:KRP}} to  $\mathrm{vec}\left\{ \ma{F}^{(3)} \diamond \Delta\ma{F}^{(1)} \right\}$, \textcolor{c44}{leading} to
\begin{align}
\mathrm{vec}\left\{\ma{F}^{(3)} \diamond  \Delta\ma{F}^{(1)} \right\}
&=\ma{G}(\ma{F}^{(3)}, M_1 )\cdot \mathrm{vec}\left\{ \Delta\ma{F}^{(1)} \right\}.\label{eq:KRP_2}
\end{align}
Furthermore, we use the results from \cref{eq:KRP_1} and \cref{eq:KRP_2}, \textcolor{c7}{as well as from  \cref{eq:SECSI_vec_DeltaF1} and \cref{eq:SECSI_vec_DeltaF3}}, in \cref{eq:vec_DeltaB} to get
\begin{align}
\mathrm{vec}\left\{ \Delta \ma{F}^{(2)} \right\} & = \ma{L}_{7} \cdot \ma{n}_1  + \mathcal{O}(\Delta^2) \, , \label{eq:SECSI_vec_DeltaF2}
\end{align}
where
\begin{align}
\ma{L}_7 &= \left( \left(\ma{F}^{(3)} \diamond \ma{F}^{(1)} \right)^{+} \otimes \ma{I}_{M_2} \right) \cdot \ma{P}_{(M_1,M_2,M_3)}^{(2)} \nonumber\\&\quad - \left( \left(\ma{F}^{(3)} \diamond \ma{F}^{(1)} \right)^{+} \otimes \ma{F}^{(2)} \right) \cdot \ma{Q}_{(M_1\cdot M_3, d)}  \nonumber \\&\quad
\cdot \left( \ma{H}(\ma{F}^{(1)}, M_3 ) \cdot \ma{L}_{5} +
\ma{G}(\ma{F}^{(3)}, M_1 ) \cdot \ma{L}_{6} \right).\label{eq:SECSI_L7}
\end{align}
By using the obtained expressions for \cref{eq:SECSI_vec_DeltaF2} in \cref{eq:SECSI_vec_rMSFEAll}, we get
\begin{align}
&\mathrm{vec}\left\{ \ma{F}^{(2)} - (\ma{F}^{(2)} +  \Delta \ma{F}^{(2)})\cdot\tilde{\ma{P}}^{(2)}\cdot\ma{P}_{\mathrm{opt}}^{(2)} \right\} \approx  \ma{L}_{\bm{F}_2} \cdot \ma{n}_1, \label{eq:SECSI_vec_rMSFE2}
\end{align}
where
\begin{align}
\ma{L}_{\bm{F}_2} &=
\left( \ma{I}_{d} \otimes \ma{F}^{(2)} \right) \cdot \left( \bm{K}_2^{-1} \otimes \ma{I}_{d} \right) \cdot \ma{W}_{(d)}\cdot \nonumber\\
&\quad \Big[ \left( \ma{I}_{d} \otimes \ma{F}^{(2)^{\mathrm{H}}} \right)\cdot \ma{L}_{7} \Big] - \ma{L}_{7} \, . \label{eq:SECSI_Lb}
\end{align}
Finally, the closed-form  expression for the \textcolor{c7}{second} factor matrix $\mathrm{rMSFE}_{\bm{F}_2}$ is approximated by
\begin{align}
\mathrm{rMSFE}_{\bm{F}_2} &= \frac{ \mathrm{tr}\left(\ma{L}_{\bm{F}_2} \cdot  \ma{R}_{\mathrm{nn}}^{(1)} \cdot \ma{L}^{\mathrm{H}}_{\bm{F}_2} \right)}{ \left\| \ma{F}^{(2)}\right\|_{\mathrm{F}}^{2}}. \label{eq:SECSI_rMSFE2}
\end{align}
Note that we have \textcolor{c7}{performed} this \textcolor{c55}{first order perturbation} analysis on the 3-mode rhs \textcolor{c7}{solution} as shown in Table \ref{alg:SECSI}, where $\mathrm{rMSFE}^{(1)} \approx \mathrm{rMSFE}_{\bm{F}_1}$, $\mathrm{rMSFE}^{(2)} \approx \mathrm{rMSFE}_{\bm{F}_2}$, and $\mathrm{rMSFE}^{(3)} \approx \mathrm{rMSFE}_{\bm{F}_3}$. Nevertheless, the 1-mode and 2-mode rhs estimates can be obtained by applying the SECSI framework on the 3-mode (\textcolor{c7}{i.e.,} Table \ref{alg:SECSI}) to the permuted versions of the input tensor $\ten{X}$. In the same manner, the $\mathrm{rMSFE}^{(r)}$ for the 1-mode and 2-mode rhs estimates can also be approximated, by applying this performance analysis framework on the permuted versions of the noiseless input tensor $\ten{X}_0$, as shown in \textcolor{c7}{Table} \ref{tab:PA_Modes}.
\begin{table}[h!]
	\centering
	\begin{tabular}{|c|c|c|c|}
		\hline
		& 3-mode & 2-mode & 1-mode \\
		\hline
		input & $\ten{X}$ & $\mathrm{perm}(\ten{X},[2,3,1])$ & $\mathrm{perm}(\ten{X},[1,3,2])$ \\
		\hline
		$\mathrm{rMSFE}^{(1)}$ & $\mathrm{rMSFE}_{\bm{F}_1}$  & $\mathrm{rMSFE}_{\bm{F}_3}$  & $\mathrm{rMSFE}_{\bm{F}_1}$ \\
		$\mathrm{rMSFE}^{(2)}$ & $\mathrm{rMSFE}_{\bm{F}_2}$  & $\mathrm{rMSFE}_{\bm{F}_1}$  & $\mathrm{rMSFE}_{\bm{F}_3}$ \\
		$\mathrm{rMSFE}^{(3)}$ & $\mathrm{rMSFE}_{\bm{F}_3}$  & $\mathrm{rMSFE}_{\bm{F}_2}$  & $\mathrm{rMSFE}_{\bm{F}_2}$ \\
		\hline
	\end{tabular}
	\caption{$\mathrm{rMSFE}^{(r)}$ from  \cref{eq:SECSI_rMSFEdef} approximation, using eqs. \eqref{eq:SECSI_rMSFE1} \eqref{eq:SECSI_rMSFE2} \eqref{eq:SECSI_rMSFE3}, for the different $r$-modes of SECSI}
	\label{tab:PA_Modes}
\end{table}

\section{Extensions}
\subsection{Extension to the 3-mode LHS}\label{sec:SECSI_LHS}
In this section, we extend the obtained results for the rhs estimates to the lhs estimates. We first describe the main differences between the rhs and lhs third mode estimates, and later redefine the performance analysis framework accordingly. Note that, since there are no significant changes for computing the lhs estimates, when compared to the rhs case, we can reuse most of the expressions obtained in Section \ref{sec:SECSI_Performance} for this extension to the lhs.

For computing the lhs estimates, the matrices $\ma{S}_{3,k}^{\mathrm{lhs}}$ from \cref{eq:LHS} are used as input to the JEVD problem, instead of the matrices $\ma{S}_{3,k}^{\mathrm{rhs}}$ from \cref{eq:RHS}. This leads to a redefinition of the matrices $\ma{K}_3^{(k)}$ and $\ma{L}_3^{(k)}$ that appear in the vectorization of $\ma{S}_{3,k}^{\mathrm{rhs}}$ (\cref{eq:SECSI_vec_Delta_S3kRHS}), in the \textcolor{c44}{rhs} case, and now are used to express  the vectorization of $\ma{S}_{3,k}^{\mathrm{lhs}}$ as \textcolor{c44}{$\mathrm{vec}\left\{ \ma{S}_{3,k}^{\mathrm{lhs}} \right\} =  \ma{L}_{3,\text{lhs}}^{(k)}\cdot \ma{n}_1$}, where
\begin{align}
\ma{L}_{3,\text{lhs}}^{(k)} &=  \bm{Q}_{d,d}^{T}\cdot
(\ma{I}_{d}\otimes \ma{S}_{3,p}^{-1}  ) \cdot  \ma{L}_{2}^{(k)} - (\ma{S}_{3,k}^{T}\cdot \ma{S}_{3,p}^{-\mathrm{T}} \otimes \ma{S}_{3,p}^{-1}) \cdot \ma{L}_{2}^{(p)}. \label{eq:SECSI_L3LHS}
\end{align}
\textcolor{c55}{The calculation of a JEVD for the slices $\hat{\ma{S}}_{3,k}^{\mathrm{lhs}}$ in \cref{eq:LHS} results in the estimation of the transformation matrix $\ma{T}_2$, instead of $\ma{T}_1$.} Therefore, this leads to a different \textcolor{c44}{way of estimating the factor matrices, as defined in Table \ref{alg:SECSI}}. For instance, $\ma{F}^{(1)}$ is estimated via a LS-fit and $\ma{F}^{(2)}$ is now estimated from $\ma{T}_2$. \textcolor{c111}{These changes lead to a redefinition of  $\ma{L}_6$ {that is similar to the definition of  $\bm{L}_7$ for the rhs (\cref{eq:SECSI_L7})}. Therefore, we refer to it as}
\begin{align}
\ma{L}_{7,\text{lhs}} &= \left( \ma{I}_{d} \otimes \ma{U}_{2}^{[\mathrm{s}]} \right) \cdot \ma{L}_{4} +  \left( \ma{T}_{2}^{\mathrm{T}} \cdot \ma{\Sigma}_2^{[\mathrm{s}]^{-1}} \ma{V}_2^{[\mathrm{s}]^{\mathrm{T}}} \otimes \ma{\Gamma}_2^{[\mathrm{n}]} \right) \, . \label{eq:SECSI_L7LHS}
\end{align}
In the same manner,  $\ma{L}_7$ for the rhs is also redefined for the lhs as
\begin{align}
\ma{L}_{6,\text{lhs}}&= \left( \left(\ma{F}^{(2)} \diamond \ma{F}^{(3)} \right)^{+} \otimes \ma{I}_{M_1} \right) \cdot \ma{P}_{(M_1,M_2,M_3)}^{(1)} \nonumber\\&\quad - \left( \left(\ma{F}^{(2)} \diamond \ma{F}^{(3)} \right)^{+} \otimes \ma{F}^{(1)} \right) \cdot \ma{Q}_{(M_2\cdot M_3, d)}  \nonumber \\&\quad
\cdot  \left( \ma{G}(\ma{F}^{(2)}, M_3 )  \cdot \ma{L}_{5} +
\ma{H}(\ma{F}^{(3)}, M_2 ) \cdot \ma{L}_{7} \right). \label{eq:SECSI_L6LHS}
\end{align}
\begin{table}[h!]
\centering
\begin{tabular}{|c|l|c|c|}
\hline
 compute & RHS              & LHS  & $\mathrm{vec}\{\cdot\}$                                                \\ \hline
 $\ma{L}_0$        &               \multicolumn{2}{c|}{\cref{eq:SECSI_vec_Delta_S}}     & $[\Delta \ten{S}^{[\mathrm{s}]}]_{(1)} $                                                 \\
 $\ma{L}_1$                       & \multicolumn{2}{c|}{ \cref{eq:SECSI_L1}}   &   $[\Delta \ten{S}_{3}^{[\mathrm{s}]}]_{(1)}$                                \\
 $\ma{L}_2^{(k)}$           & \multicolumn{2}{c|}{ \cref{eq:SECSI_L2k}} & $\Delta \ma{S}_{3,k}$                                 \\ \cline{2-3}
 $\ma{L}_3^{(k)}$           &  \cref{eq:SECSI_L3k} &  \cref{eq:SECSI_L3LHS}& $\Delta \ma{S}_{3,k}^{\mathrm{rhs}}$, $\Delta \ma{S}_{3,k}^{\mathrm{lhs}}$\\ \cline{2-3}
  $\ma{L}_3$                  &\multicolumn{2}{c|}{\cref{eq:SECSI_s}}&                                                           \\
  $\ma{L}_4^{(k)}$           & \multicolumn{2}{c|}{\cref{eq:SECSI_L4}}   &   $\Delta \ma{F}^{(3)}(k,:)$                                \\
 $\ma{L}_5$                       & \multicolumn{2}{c|}{\cref{eq:SECSI_KL5}}   &    $\Delta \ma{F}^{(3)}$                                                         \\\cline{2-3}
  $\ma{L}_6$                       & \cref{eq:SECSI_L6}   &  \cref{eq:SECSI_L6LHS}& $\Delta \ma{F}^{(1)}$\\ \cline{2-3}
 $\ma{L}_7$                       & \cref{eq:SECSI_L7}   &  \cref{eq:SECSI_L7LHS} & $\Delta \ma{F}^{(2)}$\\ \cline{2-3}
 $\ma{L}_{\bm{F}_1}$ & \multicolumn{2}{c|}{ \cref{eq:SECSI_La}}   &     \cref{eq:SECSI_vec_rMSFE1}                              \\
  $\ma{L}_{\bm{F}_2}$ & \multicolumn{2}{c|}{ \cref{eq:SECSI_Lb}}   &          \cref{eq:SECSI_vec_rMSFE2}                           \\
 $\ma{L}_{\bm{F}_3}$ & \multicolumn{2}{c|}{\cref{eq:SECSI_Lc}}     &      \cref{eq:SECSI_vec_rMSFE3}                             \\
 $\mathrm{rMSFE}^{(1)}$                       & \multicolumn{2}{c|}{$\mathrm{rMSFE}_{\bm{F}_1}$, \cref{eq:SECSI_rMSFE1}}     &                          \\
 $\mathrm{rMSFE}^{(2)}$                       & \multicolumn{2}{c|}{$\mathrm{rMSFE}_{\bm{F}_2}$,  \cref{eq:SECSI_rMSFE2}}&                              \\
 $\mathrm{rMSFE}^{(3)}$                       & \multicolumn{2}{c|}{$\mathrm{rMSFE}_{\bm{F}_3}$, \cref{eq:SECSI_rMSFE3}} &                              \\ \hline
\end{tabular}
	\caption{3-mode rhs and lhs performance analysis of SECSI. The matrix  $\ma{L}_6$ is always used for $\Delta \ma{F}^{(1)}$. Similarly  $\ma{L}_7$  is always used for $\Delta \ma{F}^{(2)}$.}
		\label{tab:PA_RHS}
\end{table}
Finally, a summary of the third mode rhs and lhs performance analysis is shown in Table \ref{tab:PA_RHS}. Note that only steps 4 ( $\ma{L}_3^{(k)}$ ), 8 ( $\ma{L}_6$), and 9 ( $\ma{L}_7$ ) are changed from the rhs to the lhs performance analysis. Moreover the lhs $\mathrm{rMSFE}^{(r)}$ expressions in the \textcolor{c55}{other} modes (i.e., 1-mode and 2-mode) can be approximated in the same manner as in Section \ref{sec:SECSI_Performance} by applying this \textcolor{c55}{first order performance analysis} to permuted versions of the noiseless tensor $\ten{X}_0$, as shown in \textcolor{c7}{Table} \ref{tab:PA_Modes}.

\subsection{Extension to the underdetermined (degenerate) case} \label{non-degenerate}
\textcolor{c55}{In this work, we have assumed the non-degenerate case, but the results can also be applied to the underdetermined (degenerate) case. The decomposition is underdetermined in mode $n$ if $d>M_n$. The SECSI framework is still applicable if the problem is underdetermined in up to one mode \textcolor{c66}{\cite{SECSI}}. For example, let \textcolor{c55}{the} tensor rank of the noiseless tensor be greater than any one of the dimension, i.e., $M_1<d \leq \text{min}\{M_2,M_3\}$. In this case, $\bm{F}^{(1)} \in \compl^{M_1 \times d}$ is a flat matrix. But $\bm{U}_1$ has the dimension of $M_1 \times M_1$. Therefore, $\bm{T}_1$ does not exist. However, the diagonalization problems $\bm{S}^\text{rhs}_{1,k}$ and $\bm{S}^\text{lhs}_{1,k}$ in Fig. \ref{overview_SECSI} can still be solved and yield two estimates for $\bm{F}^{(1)}$, one for $\bm{F}^{(2)}$, and one for $\bm{F}^{(3)}$.} \looseness=-1
\section{A Performance Analysis based Factor Matrix Estimates Selection Scheme} \label{sec:SECSI_PAS}
\begin{figure*}[t!]
	\centering
	\begin{subfigure}[b]{0.45\textwidth}
		\raggedright
		\includegraphics[width=.95\linewidth]{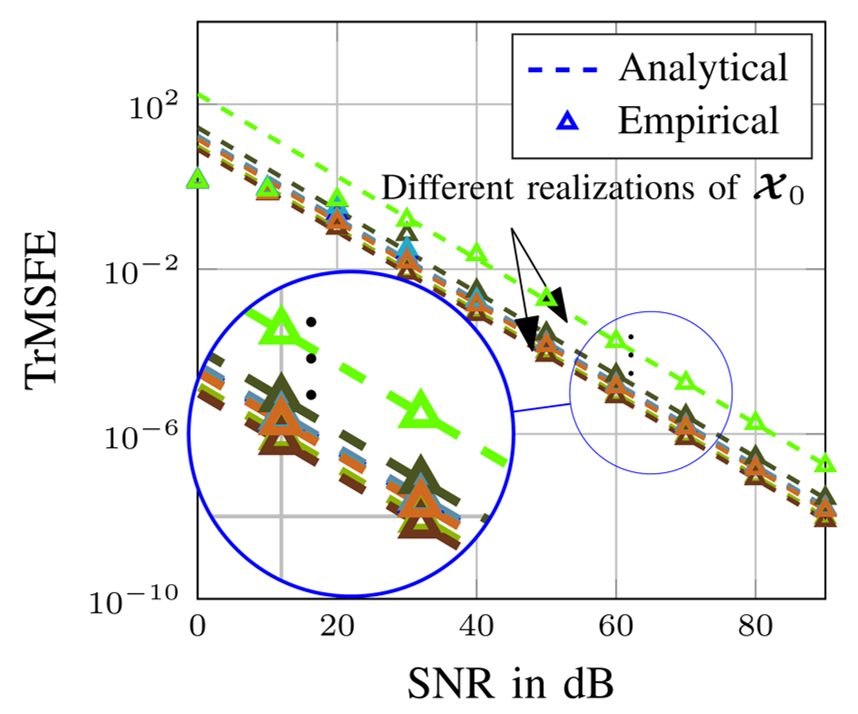}
		\caption{Scenario I}
		\label{fig:SECSI_PA_Scenario5}
	\end{subfigure}
	~
	\begin{subfigure}[b]{0.45\textwidth}
		\raggedleft
		\includegraphics[width=.9\linewidth]{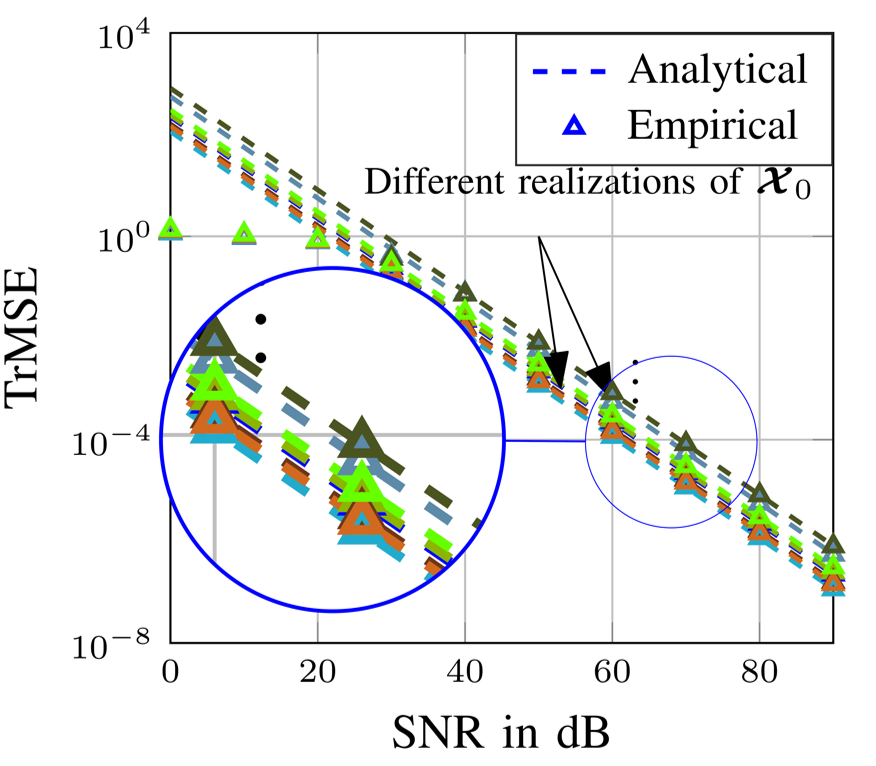}
		\caption{Scenario II}
		\label{fig:SECSI_PA_Scenario6}
	\end{subfigure}
	\caption{TrMSFE of the 3-mode rhs estimates using the SECSI framework in scenarios I and II}
	\label{fig:SECSI_PA_Scenario56}
	\vspace*{-0.5cm}
\end{figure*}
In this section we propose a new estimates selection scheme for the SECSI framework. We refer to this new scheme as Performance Analysis based Selection (PAS). Unlike other selection schemes \textcolor{c55}{for the final matrix estimates} proposed in \cite{SECSI}, such as \textcolor{c7}{CON-PS (conditioning number - paired solutions), REC-PS (reconstruction error - paired solutions), and BM (best  matching)}, which perform a heuristic selection \textcolor{c66}{or an exhaustive} search (BM) to select the \textcolor{c55}{final} estimates, this new selection scheme allows us to directly select one estimate per factor matrix (instead of performing a search among all the possible estimate combinations). \textcolor{c88}{For instance, the BM scheme tests all possible combinations of the estimates of the loading matrices in an exhaustive search. It therefore requires the reconstruction of $(R \cdot (R - 1))^R$ tensors \cite{SECSI}, which grows rapidly with $R$ and already reaches 216 for $R = 3$, whereas the PAS scheme requires only 18 estimates for $R = 3$.}

All the analytical expressions in the performance analysis are computed from the noiseless estimates (such as $\ten{X}_0$, $\ma{U}_r^{[\mathrm{s}]}$, $\ten{S}^{[\mathrm{s}]}$, etc.). For the PAS scheme, we approximate the noiseless quantities with the noisy ones as $\ten{X}_0 \approx \hat{\ten{X}}$, $\ma{U}_r^{[\mathrm{s}]} \approx \hat{\ma{U}}_r^{[\mathrm{s}]}$, $\ten{S}^{[\mathrm{s}]} \approx \hat{\ten{S}}^{[\mathrm{s}]}$, etc. Then, we use the corresponding performance analysis \textcolor{c55}{and} assume perfect knowledge of the second order moments of the noise (i.e., $\ma{R}_{\mathrm{nn}}^{(r)}$ and $\ma{C}_{\mathrm{nn}}^{(r)}$), to estimate $\mathrm{rMSFE}^{(r)}$ for all $r=1,2,3$, in all the $r$-mode lhs and rhs solutions. Finally, we select the estimates that correspond to the \textcolor{c7}{smallest} estimated $\mathrm{rMSFE}^{(r)}$ values, for all $r=1,2,3$. Note that estimating the noise variance has no effect on this scheme, since all the estimated $\mathrm{rMSFE}^{(r)}$ values are multiplied by the same $\sigma_N^2$ factor.

\section{Simulation Results}
\label{sim}
%
%
%
%
In this section, we validate the resulting analytical expressions with empirical simulations. First we compare the results for the proposed analytical framework to the empirical results. Then, we compare the performance of \textcolor{c55}{the} PAS selection scheme with other estimates selection schemes.
\vspace*{-0.5cm}
\subsection{Performance Analysis Simulations}
We define three simulation scenarios, where the properties of the noiseless tensor $\ten{X}_0$ and the number of trials used for every point of the simulation are stated in Table \ref{tab:SECSI_ScenariosSNR}.
\begin{table}[h!]
	\begin{center}
		\begin{tabular}{|c|c|c|c|c|}
			\hline
			Scenario & Size & $d$ & $\ma{F}^{(r)}$ Correlation & Trials \\
			\hline
			I & $5 \times 5 \times 5$ & 4 & (none) & $10,000$\\
			II & $5 \times 8 \times 7$ & 4 & $r=1$ & $10,000$\\
		    III & $3 \times 15 \times 70$ & 3 & (none) & $5,000$\\
			\hline
		\end{tabular}
		\caption{ Scenarios for varying SNR simulations.}
		\label{tab:SECSI_ScenariosSNR}
	\end{center}
\end{table}

%
%
\textcolor{c55}{In scenario I and scenario III, we have used real-valued tensors while complex-valued tensors are used for scenario II. Moreover, to further illustrate the robustness of our performance analysis results, we have used different JEVDs algorithms for both scenarios. We employ \textcolor{c66}{the} JDTM algorithm \cite{JDTM} for scenario I while \textcolor{c66}{the} coupled JEVD \cite{Coupled} is  employed for scenario II. Note that both algorithms are based on the indirect LS cost function \cite{Icassp_17}.}
For every trial, the noise tensor $\ten{N}$ is randomly generated and has zero-mean Gaussian entries of variance $\sigma_N^2 =  \| \ten{X}_0 \|_{\mathrm{H}}^2/ (\text{SNR}\cdot M)$. Moreover, the noiseless tensor $\ten{X}_0$  is  fixed \textcolor{c77}{in each} the experiment and has zero-mean uncorrelated Gaussian entries on its factor matrices $\ma{F}^{(r)}$. We plot several realizations of the experiments (therefore different $\ten{X}_0$) on top of each other, to  \textcolor{c55}{provide} better insights about the performance of the tested algorithms. This simulation setup is selected, since the derived analytical expression depicts the rMSFE for a known noiseless tensor $\ten{X}_0$ over several noise trials. Moreover, every realization of the noiseless tensor $\ten{X}_0$ is \textcolor{c77}{given by}  $\ten{X}_0 = \ten{I}_{3,d} \times_1 \ma{F}^{(1)} \times_2 \ma{F}^{(2)} \times_3 \ma{F}^{(3)}$, where the factor matrices $\ma{F}^{(r)} \in \real^{M_r \times d}$ have uncorrelated Gaussian entries for all $r=1,2,3$ for scenario I. 
Furthermore,  for scenario II, we also introduce correlation in the factor matrix $\ma{F}^{(1)}$. In this scenario, the factor matrices $\ma{F}^{(2)}$ and $\ma{F}^{(3)}$ are randomly drawn but $\ma{F}^{(1)}$ is fixed along the experiments as
\begin{equation}\label{eq:SECSI_Scenario5_F1}
\ma{F}^{(1)} =
\begin{bmatrix}
1		&1		&1		&1\\
1		&0.95	&0.95	&0.95\\
1		&0.95	&1		&1\\
1		&1		&0.95	&1\\
0.95	&1		&1		&1
\end{bmatrix}\, .
\end{equation}
We \textcolor{c55}{depict the} results in the form of \textcolor{c55}{the} Total rMSFE (TrMSFE) since it reflects the total factor matrix estimation accuracy of the tested algorithms. The TrMSFE is defined as $\mathrm{TrMSFE} = \sum_{r=1}^{3} \mathrm{rMSFE}^{(r)}$
%
where $\mathrm{rMSFE}^{(r)}$ is the same as in \cref{eq:SECSI_rMSFEdef}.

\textcolor{c55}{It is evident from the results in Fig. \ref{fig:SECSI_PA_Scenario56} that the analytical results from the proposed first-order perturbation analysis match well with the Monte-Carlo simulations for both real and complex valued tensors. Moreover, the results also \textcolor{c66}{show} an excellent match to the empirical simulations for scenario II where we have an asymmetric tensor and also have high correlation in the first factor matrix.}
\begin{figure*}[t!]
	\centering
	\begin{subfigure}[t]{0.27\textwidth}
		\raggedleft
		\includegraphics[width=.95\linewidth]{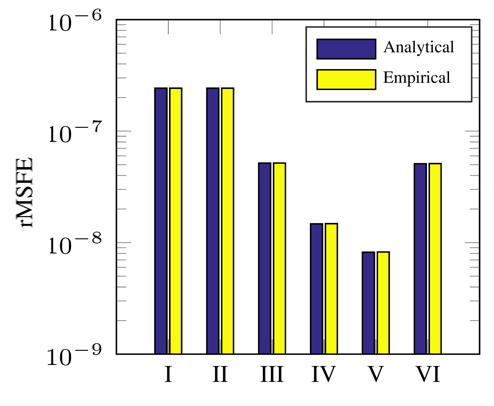}
		\caption{6 estimates of $\bm{F}^{(1)}$}
		\label{fig:6estF1}
	\end{subfigure}
	\begin{subfigure}[t]{0.27\textwidth}
		\raggedleft
		\includegraphics[width=.95\linewidth]{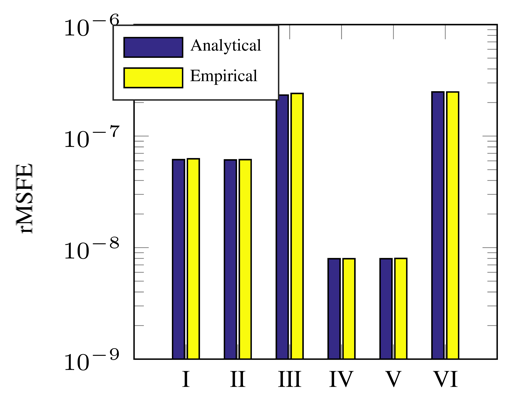}
		\caption{6 estimates of $\bm{F}^{(2)}$}
		\label{fig:6estF2}
	\end{subfigure}
		\begin{subfigure}[t]{0.27\textwidth}
			\raggedleft
			\includegraphics[width=.95\linewidth]{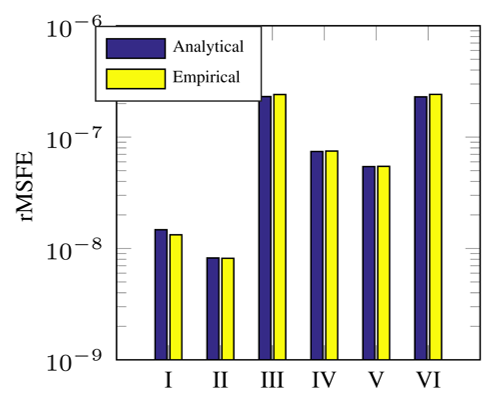}
			\caption{6 estimates of $\bm{F}^{(3)}$}
			\label{fig:6estF3}
		\end{subfigure}
	\caption{6 estimates for each of the factor matrix for scenario III and for a SNR = 50 dB.}\label{fig:6estF}
\end{figure*}

In Fig. \ref{fig:6estF}, we show the analytical and empirical results for the asymmetric scenario III where the results are shown for 6 estimates of each factor matrix obtained by solving all of the $r$-modes for a SNR = 50 dB. The estimates are arranged according to Fig. \ref{overview_SECSI}. The results show an excellent match for emiprical and analytical results for all of 6 estimates for each factor matrix. Moreover, the estimates $\bm{F}_\mathrm{IV}^{(1)}$, $\bm{F}_\mathrm{V}^{(1)}$, $\bm{F}_\mathrm{IV}^{(2)}$, $\bm{F}_\mathrm{V}^{(2)}$, $\bm{F}_\mathrm{I}^{(3)}$, and $\bm{F}_\mathrm{II}^{(3)}$ are the best estimates for each of the factor matrices. Note that all of these estimates are obtained by solving the JEVD problem  for \textcolor{c77}{the 3-mode in \cref{eq:RHS} and \cref{eq:LHS}. This results from the fact that JEVD for this mode has the highest number of slices (70) which results in a better accuracy.} \looseness=-1

\subsection{SECSI framework based on PAS scheme }
In this section, we compare the \textcolor{c55}{performance} of the PAS scheme \textcolor{c55}{with other selection schemes}, such as the CON-PS, REC-PS, and BM schemes from \cite{SECSI} for two simulation scenarios in Table \ref{tab:SECSI_ScenariosFixed}.
\begin{table}[h!]
	\begin{center}
		\begin{tabular}{|c|c|c|c|c|c|}
			\hline
			Scenario & Size & SNR & $d$ & $\ma{F}^{(r)}$ Correlation & Trials \\
			\hline
		   IV & $5 \times 5 \times 5$ & 50 dB & 3 & (none) & $10,000$\\
			V & $3 \times 15 \times 70$ & 40 dB & 3 & (none) & $5,000$\\
			\hline
		\end{tabular}
		\caption{ Scenarios for fixed SNR simulations.}
		\label{tab:SECSI_ScenariosFixed}
	\end{center}
\end{table}

In these simulation setups, the noise tensor $\ten{N}$ as well as the noiseless tensor $\ten{X}_0$ are randomly generated at every trial. \textcolor{c55}{Since the noise variance has no effect on \textcolor{c66}{the selected estimate}, for the real-valued uncorrelated noise with equal variance scenario, we use $\ma{R}_{\mathrm{nn}}^{(r)} = \ma{I}_{M}$ and $\ma{C}_{\mathrm{nn}}^{(r)} = \ma{I}_{M}$, for all $r=1,2,3$, as input for this PAS scheme, regardless of the value of $\sigma_N^2$}. Moreover, to provide \textcolor{c55}{more} insights about the performance, we also define a naive selection scheme denoted as DUMMYR, where the final estimates are randomly selected among the 6 possible estimates available per factor matrix. We use the complementary cumulative density function (CCDF) of the TrMSFE to \textcolor{c55}{illustrate} the robustness of the \textcolor{c55}{selected strategy} \cite{SECSI}. \looseness=-1
\begin{figure*}[t!]
	\centering
	\begin{subfigure}[t]{0.45\textwidth}
		\raggedleft
		\includegraphics[width=.95\linewidth]{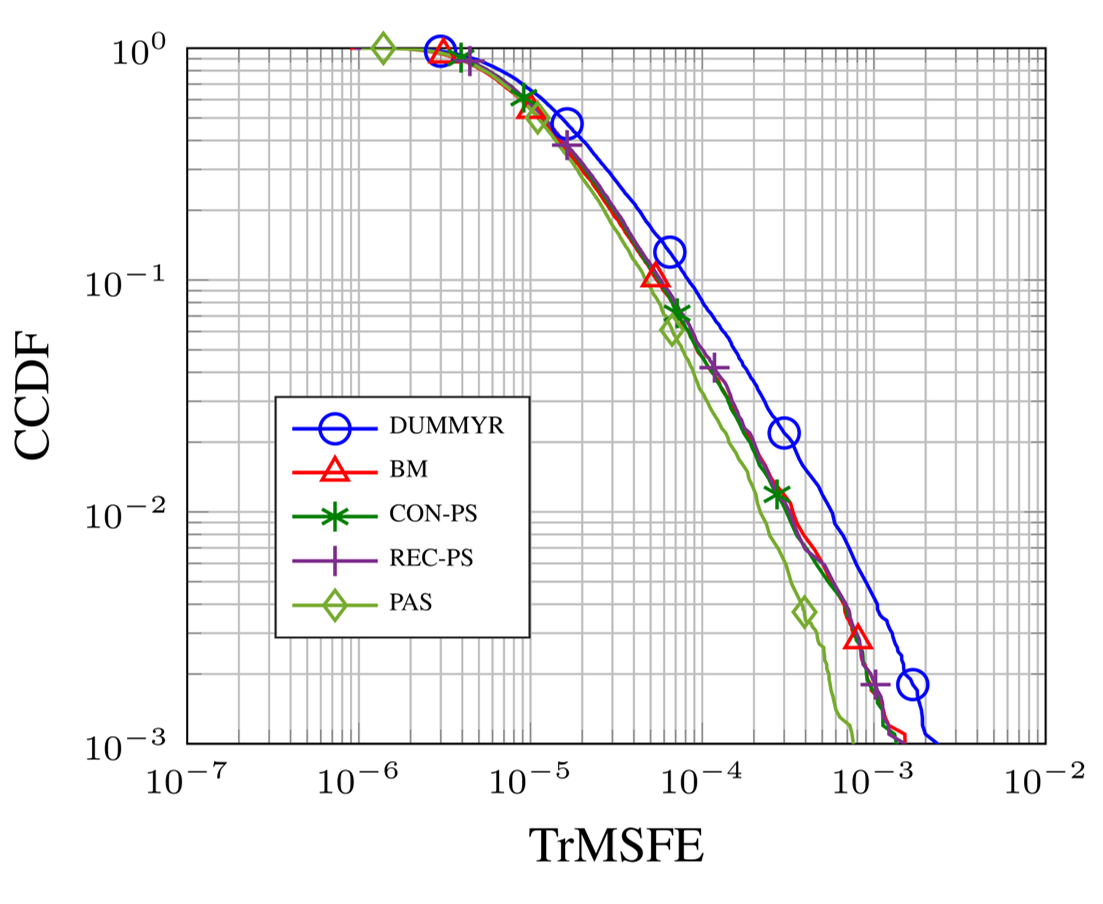}
		\caption{Scenario IV}
		\label{fig:CCDF_2}
	\end{subfigure}
	\begin{subfigure}[t]{0.45\textwidth}
		\raggedleft
		\includegraphics[width=.95\linewidth]{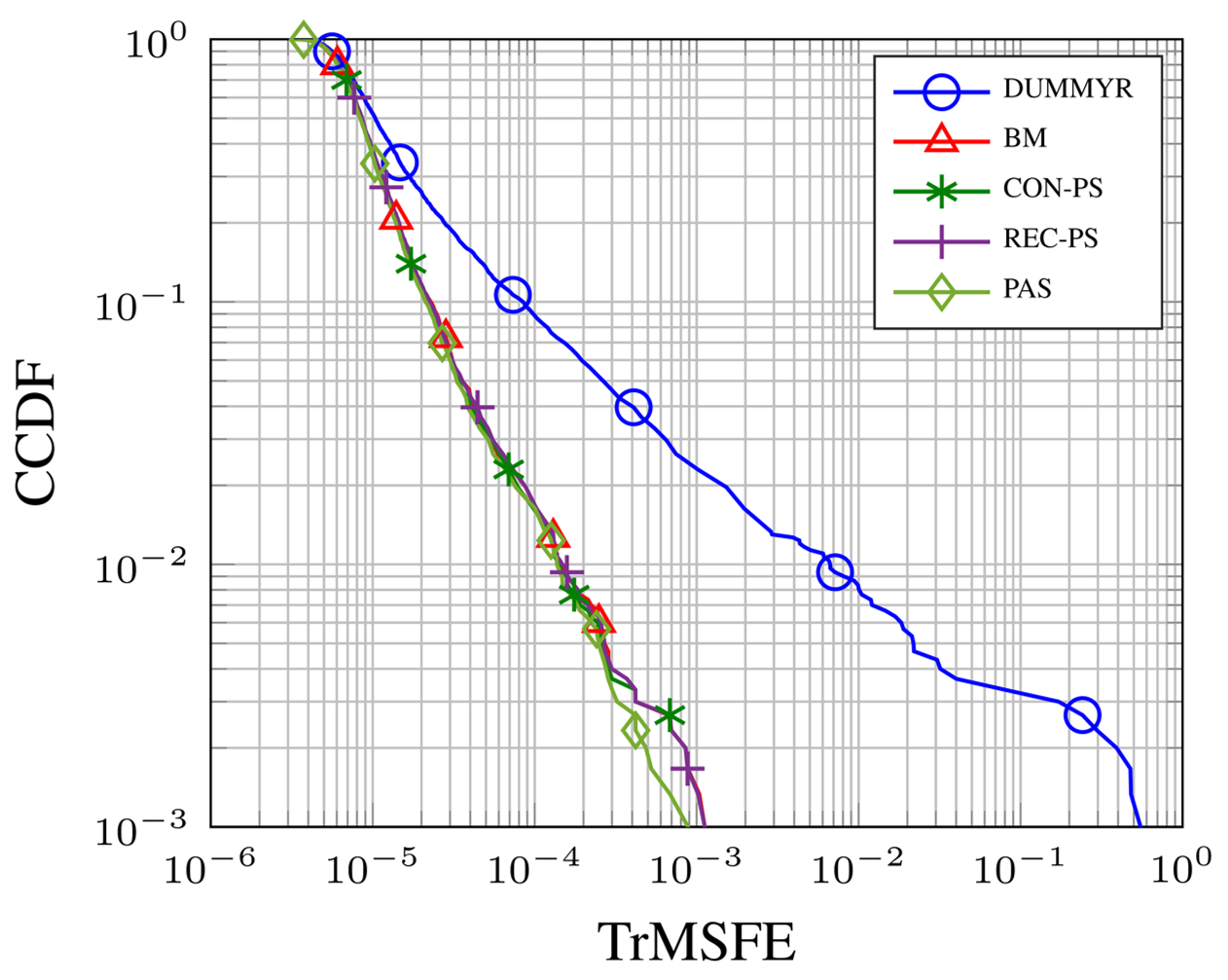}
		\caption{Scenario V}
		\label{fig:CCDF_4}
	\end{subfigure}
	\caption{CCDF of the TrMSFE for Scenarios IV and V.}\label{fig:CCDF}

\end{figure*}

The results are shown in Fig. \ref{fig:CCDF} for two scenarios with  different tensor sizes. The results show that the proposed PAS scheme outperforms the other schemes especially in scenario IV where the tensor size is small and \textcolor{c55}{the} SNR is also high. This is to be expected, since the expressions used to build the PAS scheme are based on the high SNR assumption. Nevertheless, we also observe that even with the higher dimensional tensor and relatively low SNR, as shown in Fig. \ref{fig:CCDF}(b), we observe that the PAS scheme performs slightly better than the CON-PS, REC-PS, and BM schemes.
%
\section{Conclusions}
\label{conc}
\textcolor{c55}{In this work, a first-order perturbation analysis of the SECSI framework for \textcolor{c66}{the approximate CP decomposition of 3-D noise-corrupted low-rank tensors is presented}, where we {provide} {closed-form} expressions for the {relative mean square error} for each of the estimated factor matrices. {The derived expressions are formulated in terms of the second-order moments of the noise, such that apart from a zero mean, no assumptions on the noise statistics are required}. The simulation results depict the excellent match between the \mbox{{closed-form}} expressions and the empirical results for both real and complex valued data. Moreover, these expressions can {also be} used for an enhancement in the factor matrix estimates selection step of the SECSI framework.  The new PAS selection scheme outperforms the existing schemes especially at high SNR values.}
%
%
%

\appendices
 \section{Proof of Theorem \ref{teo:Popt} }\label{proof:Popt}\textcolor{c88}
{The term $(\ma{Z} + \Delta \ma{Z})\cdot \tilde{\ma{P}}$ can be interpreted as the estimate of $\ma{Z}$ up to a scaling of its column, given by
 \begin{align*}
 \hat{\ma{Z}} =(\ma{Z} + \Delta \ma{Z})\cdot\tilde{\ma{P}} = \tilde{\ma{Z}} +\Delta\tilde{\ma{Z}}.
 \end{align*}
 Therefore, $\tilde{\ma{P}}$ can be expressed as $ \tilde{\ma{P}} =\mathrm{Ddiag}\left( \ma{Z}^{\mathrm{H}} \cdot \tilde{\ma{Z}} \cdot \bm{K}^{-1} \right)$.
 To resolve the scaling ambiguity, we compute a diagonal matrix $\ma{P}_{\mathrm{opt}}$ that satisfies ${\ma{P}_{\mathrm{opt}}}=\mathrm{Ddiag}( \ma{Z}^{\mathrm{H}} \cdot \hat{\ma{Z}}\cdot \bm{K}^{-1} )^{-1}$. This result leads to \looseness=-1
 \begin{align*}
 \tilde{\ma{P}}\cdot \ma{P}_{\mathrm{opt}}&= \tilde{\ma{P}} \cdot\mathrm{Ddiag}\left( \ma{Z}^{\mathrm{H}} \cdot \hat{\ma{Z}} \cdot \bm{K}^{-1} \right)^{-1}\\
 &=\tilde{\ma{P}} \cdot \mathrm{Ddiag}\left( \ma{Z}^{\mathrm{H}}\cdot (\ma{Z} + \Delta \ma{Z})\cdot\tilde{\ma{P}} \cdot \bm{K}^{-1} \right)^{-1} \\
  &= \left(  \ma{I}_d + \mathrm{Ddiag}(\ma{Z}^{\mathrm{H}}\cdot\Delta \ma{Z}) \cdot \bm{K}^{-1} \right)^{-1}
 \end{align*}
Let us define the diagonal matrix $\ma{U} \triangleq \mathrm{Ddiag}\left(\ma{Z}^{\mathrm{H}}\cdot\Delta \ma{Z} \right)\cdot \bm{K}^{-1}$ of size $d \times d$ with diagonal elements $\ma{U}(i,i)$ for all $i=1,2,\dots,d$. Therefore, $\left(\ma{I}_d + \ma{U}\right)^{-1}$ is also a diagonal matrix, with diagonal elements $(1 + \ma{U}(i,i))^{-1}$ for all $i=1,2,\dots,d$. Since the elements $\ma{U}(i,i)$ are first-order terms, we use the well known Taylor \textcolor{c44}{expansion} $(1 + x)^{-1} = 1 - x + \mathcal{O}(\Delta^2) $
which holds for any first-order scalar term $x$. This leads to $\left(\ma{I}_d + \ma{U}\right)^{-1} = \left(\ma{I}_d -\ma{U}\right)  + \mathcal{O}(\Delta^2)$. Therefore, we get $\left(  \ma{I}_d + \mathrm{Ddiag}(\ma{Z}^{\mathrm{H}}\cdot\Delta \ma{Z}) \cdot \bm{K}^{-1}\right)^{-1} = \ma{I}_{d} -  \mathrm{Ddiag}\left( \ma{Z}^{\mathrm{H}} \cdot  \Delta\ma{Z} \right)\cdot \bm{K}^{-1} + \mathcal{O}(\Delta^2)$, and thus
\begin{align*}
\ma{Z}- &(\ma{Z} + \Delta \ma{Z}) \cdot \tilde{\ma{P}} \cdot\ma{P}_{\mathrm{opt}}
= \ma{Z}-(\ma{Z} + \Delta \ma{Z})\cdot\\ &\left[  \ma{I}_{d} -  \mathrm{Ddiag}\left( \ma{Z}^{\mathrm{H}} \cdot  \Delta\ma{Z} \right) \cdot \bm{K}^{-1} \right] +\mathcal{O}(\Delta^2)\\
&= \ma{Z} \cdot \mathrm{Ddiag}\left(\ma{Z}^{\mathrm{H}} \cdot  \Delta\ma{Z} \right)\cdot \bm{K}^{-1} -\Delta\ma{Z} +\mathcal{O}(\Delta^2),
\end{align*}
which proves this theorem.}

{\section*{ACKNOWLEDGMENTS}}
The  authors gratefully acknowledge the financial support by the	German-Israeli foundation (GIF), grant number I-1282-406.10/2014.

\bibliographystyle{IEEEtran}
\bibliography{mybibfile}
\end{document}